\documentclass[journal,onecolumn,letterpaper,12pt]{IEEEtran}

\usepackage{makecell}
\usepackage[utf8]{inputenc} 
\usepackage[T1]{fontenc}
\usepackage{url}
\usepackage{ifthen}
\usepackage[cmex10]{amsmath} 

\usepackage{amssymb,amsfonts}
\usepackage[dvipsnames,table,xcdraw]{xcolor}
 \usepackage[colorlinks = true, linkcolor = blue, urlcolor=blue, citecolor = red, anchorcolor = green]{hyperref}

\usepackage{multicol}
\usepackage{graphicx} 

\usepackage{amsthm}
\usepackage{thmtools}
\usepackage[capitalize]{cleveref}

\newtheorem{theorem}{Theorem}
\newtheorem{lemma}[theorem]{Lemma}

\newtheorem{definition}[theorem]{Definition}

\newtheorem{remark}[theorem]{Remark}
\newtheorem{proposition}[theorem]{Proposition}

\usepackage{mathtools}
\usepackage{bbm}
\usepackage{braket}

\DeclareMathOperator{\Tr}{Tr}
\newcommand{\tr}{\Tr}
\newcommand{\cA}{\mathcal{A}}
\newcommand{\cB}{\mathcal{B}}
\newcommand{\cC}{\mathcal{C}}
\newcommand{\cD}{\mathcal{D}}

\newcommand{\cH}{\mathcal{H}}

\newcommand{\cK}{\mathcal{K}}
\newcommand{\cL}{\mathcal{L}}
\newcommand{\cM}{\mathcal{M}}
\newcommand{\cN}{\mathcal{N}}

\newcommand{\cP}{\mathcal{P}}

\newcommand{\cR}{\mathcal{R}}
\newcommand{\cS}{\mathcal{S}}

\newcommand{\cU}{\mathcal{U}}
\newcommand{\cV}{\mathcal{V}}

\newcommand{\cY}{\mathcal{Y}}

\newcommand{\Ad}[1]{\mathrm{Ad}_{#1}}

\allowdisplaybreaks

\newcommand{\changing}[1] {{\color{black} #1}}

\usepackage[backend=bibtex8,sorting=nty,giveninits=true,doi=false,isbn=false,url=false,maxbibnames=10]{biblatex}
\renewbibmacro{in:}{}
\newbibmacro{string+doi}[1]{\iffieldundef{doi}{#1}{\href{https://dx.doi.org/\thefield{doi}}{#1}}}
\DeclareFieldFormat{title}{\usebibmacro{string+doi}{\mkbibemph{#1}}}
\DeclareFieldFormat[article]{title}{\usebibmacro{string+doi}{\mkbibquote{#1}}}
\DeclareFieldFormat[incollection]{title}{\usebibmacro{string+doi}{\mkbibquote{#1}}}                   
\DeclareFieldFormat[inproceedings]{title}{\usebibmacro{string+doi}{\mkbibquote{#1}}} 

\begin{document}

\title{Privacy-Utility Tradeoffs with \\ Quantum Local Differential Privacy}

\author{
\IEEEauthorblockN{
Theshani Nuradha\IEEEauthorrefmark{1}\IEEEauthorrefmark{2}
\thanks{{nuradha@illinois.edu, theshani.gallage@gmail.com (Corresponding author: T. Nuradha)}},
Sujeet Bhalerao\IEEEauthorrefmark{1}
,
Felix Leditzky\IEEEauthorrefmark{1}\IEEEauthorrefmark{2}
}

\IEEEauthorblockA{\IEEEauthorrefmark{1}
\small Department of Mathematics,
University of Illinois Urbana-Champaign,
Urbana, IL 61801, USA
}

\IEEEauthorblockA{\IEEEauthorrefmark{2}
\small Illinois Quantum Information Science and Technology (IQUIST) Center,\\
University of Illinois Urbana-Champaign,
Urbana, IL 61801, USA
}
}

\maketitle

\begin{abstract}%
When sensitive information is encoded in data, it is important to ensure the privacy of information when attempting to learn useful information from the data. There is a natural tradeoff whereby increasing privacy requirements may decrease the utility of a learning protocol. In the quantum setting of differential privacy, such tradeoffs between privacy and utility have so far remained largely unexplored. In this work, we study optimal privacy-utility tradeoffs for both generic and application-specific utility metrics when privacy is quantified by $(\varepsilon,\delta)$-quantum local differential privacy. In the generic setting, we focus on optimizing fidelity and trace distance between the original state and the privatized state. We show that the depolarizing mechanism achieves the optimal utility for given privacy requirements. We then study the specific application of learning the expectation of an observable with respect to an input state when only given access to privatized states. We derive a lower bound on the number of samples of privatized data required to achieve a fixed accuracy guarantee with high probability. To prove this result, we employ existing lower bounds on private quantum hypothesis testing, thus showcasing the first operational use of them. We also devise private mechanisms that achieve optimal sample complexity with respect to the privacy parameters and accuracy parameters, demonstrating that utility can be significantly improved for specific tasks in contrast to the generic setting. In addition, we show that the number of samples required to privately learn observable expectation values scales as $\Theta((\varepsilon \beta)^{-2})$, where $\varepsilon \in (0,1)$ is the privacy parameter and $\beta$ is the accuracy tolerance. We conclude by initiating the study of a private version of classical shadows to estimate the expectation of multiple observables privately, which promises useful applications for private learning tasks.
\end{abstract}

\tableofcontents

\section{Introduction}

With the exponential growth of data generated and shared through modern data systems, it is important to ensure the privacy of data within the processes of data-driven decision-making. To ensure privacy for classical data, statistical privacy frameworks have been introduced. Differential privacy (DP) is one of the widely used frameworks that ensures the privacy of individual data while answering aggregate queries \cite{DMNS06,DR14}. 
With the development of DP, various variants of DP have also been introduced~\cite{mironov2017renyi, CY16, KM14,SoK_DP22}. 
Local DP is one of such variants, where data is privatized at the local devices before releasing the data for the central data-processing \cite{erlingsson2014rappor,kairouz2014extremal}. With that, it is important to study how one could learn useful information through statistical inference with access to privatized data instead of original data. 
Previous works have studied the cost associated with these private learning processes~\cite{kamath2019privately,kamath2020private,structure_optimal_16,Ankit2023simple,Contraction_local_new24,duchi2018minimax,DJW13,DP_testing18,nuradha2022pufferfishJ}. In fact, privacy-utility tradeoffs are observed in many classical inference tasks, especially in the high-privacy regime
~\cite{rassouli2019optimalPrivacyUtilityClassical,SKT18,wang2017privacy,zhao2020notAccuracyLoss,asoodeh2021local,murakami2019utilityOptimizedLDP}.  

With the rapid growth of quantum technologies and the improvement of the capabilities of quantum devices, it is becoming increasingly important to consider the privacy of the data generated from those systems or fed into those systems. {In particular, privacy for quantum systems and data becomes important in applications such as quantum sensing (physical parameters like magnetic fields, location-specific details intrinsically encoded into quantum states), quantum machine learning models (classical data encoded into quantum circuits), and distributed quantum systems (multiple quantum parties transmitting data to another party in a network).}
Towards this goal, quantum DP (QDP) frameworks and various operational variants inspired by the classical DP framework have been introduced and studied~\cite{QDP_computation17, aaronson2019gentle, hirche2023quantum, nuradha_QPP,angrisani2023unifying,measuredHS25,gallage2025theory,dasgupta2025quantum,alabi2026QDPComposition}. QDP ensures that two distinct states considered as neighbors are hard to distinguish after being processed by a quantum channel modeling a private mechanism. Quantum local DP (QLDP) is the local variant of QDP, ensuring that two distinct quantum 
states passed through a quantum channel are hard to distinguish by a measurement (see~\eqref{eq:QLDP-def}; all distinct states are neighbors in this setting). This privacy framework is suitable for applications with several local parties generating the states, such as learning about specific properties of the quantum states while not learning the exact quantum state contributed by the local parties.

Similar to the classical setting, privacy in the quantum setting often comes with a cost to pay. For instance, one can make the system perfectly private by using a replacement channel as a private mechanism. However, the information about the original data will be completely lost and impossible to recover after this process, leading to zero utility from the process. Thus, there exist tradeoffs to be understood when providing privacy guarantees~\cite{guan2024optimal,nuradha_QPP,nuradha2024contraction,Christoph2024sample,Farhad_QP_HT,angrisani_localModel25}. Previous studies have focused on generic utility optimization under assumptions on the private mechanism, e.g., unital channels as private mechanisms~\cite{guan2024optimal,nuradha_QPP}, and on specialized applications related to testing such as quantum hypothesis testing~\cite{nuradha2024contraction,Christoph2024sample,Farhad_QP_HT,namPUT_HT_2025quantum}. However, fundamental questions on the optimality of private mechanisms and the cost one needs to pay in ensuring privacy in generic and application-specific scenarios are not well known.

In this work, we study generic utility optimization while ensuring QLDP, and obtain optimal private mechanisms without assuming any special properties of them. We also study an important use case of QLDP in estimating the expectation value of an observable given access to privatized quantum states. To this end, we formulate a private version of the classical shadow framework introduced in~\cite{huang2020predicting,huang2022learning} and utilize its noisy variant~\cite{Koh2022classicalshadows}, enabling a broad range of applications for privately estimating properties of quantum states.  In this estimation setting, we argue that application-specific utility measures are needed to thoroughly study privacy-utility tradeoffs whenever we are aware of how our data is used for inference tasks, whether fully or partially. This work also serves both as a foundation and an invitation to study statistical tasks in quantum information processing while ensuring privacy for quantum data and both fully quantum and hybrid systems.

\subsection{Main Contributions}

\noindent \textbf{Generic Utility Measures:} When one has access to privatized state $\cA(\rho)$ with $\cA$ satisfying $(\varepsilon,\delta)$-QLDP, we provide optimal utility that we can achieve when utility is defined with respect to fidelity and trace distance (\Cref{def:fidelity_utility})  in \Cref{thm:optimal-utility} and~\eqref{eq:utility_optimized_post_processing}. Note that these generic utility measures quantify the similarity or the closeness of the privatized state and the original state that holds for all input states (in fact, optimizing the worst-case scenario). In our work, we study the general setting of $(\varepsilon,\delta)$-QLDP and all private mechanisms $\cA$, showcasing that the optimal utility is achieved by a depolarizing channel with the flipping parameter $p=d(1-\delta)/(e^\varepsilon +d-1)$ together with \Cref{prop:optimal_P_QLDP_Dep}. With that, our results recover the previously known optimal utility results under $(\varepsilon,0)$-QLDP for unital channels in~\cite{guan2024optimal}. One of the main technical tools in obtaining the optimality result is that if a channel $\cN$ satisfies $(\varepsilon,\delta)$-QLDP, the average channel $\cN_G$ (twirled channel with respect to the Haar measure on a compact group) also satisfies $(\varepsilon,\delta)$-QLDP (\Cref{prop:N_u_N_G_private}).
    Note that these utility measures are generic and could be useful in situations where we are not aware of how these data are used in the privatization stage, and the utility is measured with respect to having access to one copy of the state.

\medskip
    \noindent \textbf{Application-Specific Utility Measures:}  Next, we study an operationally motivated setting for estimating the expectation of an observable with respect to an input state with access to only privatized copies of the state (estimating $\Tr[O \rho]$ with access to $n$ copies of $\cA(\rho)$ with $\cA$ being a private mechanism). In this task, we define utility as the sample complexity to achieve a required accuracy guarantee as formally defined in~\eqref{eq:learning_exp_SC_def}. We obtain lower bounds on the defined sample complexity in \Cref{prop:LB_1} by utilizing the sample complexity of private quantum hypothesis testing for $(\varepsilon,0)$-QLDP, showcasing the first operational use of those characterizations in obtaining limitations of private learning in the quantum setting. In the high privacy regime (i.e.; $\varepsilon \leq 1$), the lower bound on the sample complexity of estimation scales in the order of $1/\varepsilon^2$. 
    
    Furthermore, we provide achievability protocols ($\cA$- private mechanism and $\cP_O$-observable estimation protocol) in \Cref{prop:upper_1} to obtain upper bounds. Our first approach is to do Pauli measurements randomly based on a probability distribution depending on the observable $O= \sum_P \alpha_P P$, where each $P$ is an $n$-qubit Pauli observable, and apply a depolarizing channel after that.
    With the upper and lower bounds, we show that for $\varepsilon \leq 1$ and fixed observable $O$, the optimal sample complexity scales $\Theta\!\left( \frac{\ln(1/\eta)}{ \varepsilon^2 \beta^2} \right)$, where $\beta,\eta$ are parameters governing the probabilistic accuracy demands (\Cref{prop:optimal_SC_O}). {Moreover,  we study a private version of classical shadows by utilizing the noisy variant of classical shadows (\Cref{prop:upper_2}) and appropriately choosing the parameters of the noisy channel. This approach scales well with an increasing number of observables ($\sim\log(M)$ in \Cref{prop:manyObservable}) and works well even if one does not have an idea about the observable at the privatization stage {at the expense of the scaling with the dimension}.}
    Finally, we discuss the communication costs of different approaches and the dependence of privacy parameters on sample complexity bounds in the high privacy regimes.

\section{Notation and Preliminaries}
\label{Sec:Notation_Preliminiaries}

\subsection{Notations}

A quantum system $R$ is identified with a finite-dimensional Hilbert space~$\cH_R$. We denote the set of linear operators acting on $\cH_R$ by $\cL(\cH_R)$. 
Let $\Tr\!\left[C \right]$ denote the trace of $C$. 
 The trace norm  of a matrix $B$ is defined as $\left\|B\right\|_1 \coloneqq \Tr[\sqrt{B^\dagger B} ]$. For operators $A$ and $B$, the notation $A \geq B$ indicates that $A-B$ is a positive semi-definite (PSD) operator, while $A > B$ indicates that $A-B$ is a positive definite operator.
A quantum state $\rho_R\in\cL(\cH_R)$ on $R$ is a PSD, unit-trace operator acting on $\cH_R$. We denote the set of all quantum states in 
$\cL(\cH_R)$ by $\cD(\cH_R)$. 
A state $\rho_R$ of rank one is called pure, and we may choose a normalized vector $| \psi \rangle \in \cH_R$ satisfying $\rho_R= | \psi \rangle\!\langle \psi | $ in this case. Otherwise,
$\rho_R$ is called a mixed state. 
A quantum channel $\cN \colon \cL(\cH_A ) \to \cL(\cH_B)$ is a linear, 
completely positive, and trace-preserving (CPTP) map from $\cL(\cH_A)$ to $\cL(\cH_B)$. We often use the shorthand $\mathcal{N}_{A\to B}$ for such a channel, and we use the notation CPTP to denote the set of all quantum channels. 
A measurement of a quantum system $R$ is described by a
positive operator-valued measure (POVM) $\{M_y\}_{y \in \cY}$, which is defined to be a collection of PSD operators  satisfying $\sum_{y \in \cY} M_y= I_{R}$, where $\cY$ is a finite alphabet. According to the Born rule, after applying the above POVM to $\rho \in \cD(\cH_R)$, the probability of observing the outcome $y$ is given by~$\Tr\!\left[M_y \rho \right]$.

For $\sigma$ a state, the normalized trace distance between the states $\rho$ and $\sigma$ is defined as $ T(\rho,\sigma)\coloneqq \frac{1}{2} \left\| \rho -\sigma\right\|_1.$
It generalizes the total-variation distance between two probability distributions. For $\gamma \geq 1$,
the quantum hockey-stick divergence is defined as~\cite{SW12}
\begin{equation}\label{eq:hockey_stick}
    E_\gamma(\rho \Vert \sigma) \coloneqq \Tr\!\left[(\rho -\gamma \sigma)_{+} \right],
\end{equation}
where $\left(  A\right)  _{+}\coloneqq \sum_{i\colon a_{i}\geq0}a_{i}|i\rangle\!\langle i| $
for a Hermitian operator $A = \sum_{i}a_{i}|i\rangle\!\langle i|$.  For $\gamma=1$, observe that $ E_1(\rho \Vert \sigma)=T(\rho,\sigma)$.
For a real number $x$, we set $(x)_+ = \max\{x,0\}.$
The Uhlmann fidelity is defined as~\cite{Uhl76_nourl}
\begin{equation}\label{eq:fidelity}
F(\rho,\sigma) \coloneqq \left\| \sqrt{\rho} \sqrt{\sigma}\right\|_1^2.
\end{equation}

\begin{definition} \label{def:twirled_channel}
    Let $G$ be a compact group with unitary representations $g\mapsto U_A(g)$ on $A$ and $g\mapsto V_B(g)$ on $B$.
For a unitary $W$, write $\Ad{W}$ for the adjoint representation $\Ad{W}(X)\coloneqq W X W^\dagger$. Given any channel $\mathcal N\colon A\to B$, define its $G$-\emph{twirl} by $\mathcal N_G \ \coloneqq \ \int_G \Ad{V_B(g)^\dagger} \circ \mathcal N \circ \Ad{U_A(g)} \mathrm d\mu(g) ,$
where $\mu$ is the Haar measure on $G$. 
\end{definition}

\subsection{Quantum Local Differential Privacy}

Quantum local differential privacy (QLDP) ensures the privacy of quantum states, which are given as inputs to a  private quantum channel~\cite{hirche2023quantum,nuradha_QPP}. 
\begin{definition}[Quantum Local Differential Privacy]\label{def:QLDP}
    Fix $\varepsilon \geq 0$ and $\delta \in [0,1]$. 
    Let $\cA$ be a quantum algorithm (viz., a quantum channel). The algorithm~$\cA$ is $ (\varepsilon, \delta)$-quantum local differentially private (or $(\varepsilon,\delta)$-QLDP) if 
\begin{equation} \Tr\!\left[M \cA(\rho)\right] \leq e^\varepsilon \Tr\!\left[M \cA(\sigma)\right] + \delta,\qquad \forall  \rho, \sigma \in \cD(\cH), \quad \forall M: 0\leq M \leq I.
\label{eq:QLDP-def}
\end{equation}
We say that $\cA$ satisfies $\varepsilon$-QLDP if it satisfies $(\varepsilon,0)$-QLDP.
\end{definition}
Note that we are using the names mechanism, channel, and algorithm interchangeably to identify $\cA$ that satisfies $(\varepsilon,\delta)$-QLDP. It was shown in~\cite[Eq.~(V.1)]{hirche2023quantum} that  $(\varepsilon,\delta)$-QLDP of
a mechanism $\cA$ is equivalent to the condition
\begin{equation}
\label{eq:equivalent_HS_QLDP_delta}
    \sup_{\rho,\sigma \in \cD(\cH)}  E_{e^\varepsilon}\!\left( \cA(\rho) \Vert \cA(\sigma) \right)  \leq \delta,
\end{equation}
 where $E_{\gamma}(\cdot \| \cdot )$ is defined in~\eqref{eq:hockey_stick} for $\gamma \geq 1$. Furthermore, it is sufficient to consider the supremum over orthogonal pure states in~\eqref{eq:equivalent_HS_QLDP_delta}.
 That is, $(\varepsilon,\delta)$-QLDP is equivalent to
\begin{equation} \label{eq:equivalence_pure}
    \sup_{\rho,\sigma \in \cD(\cH)}  E_{e^\varepsilon}\!\left( \cA(\rho) \Vert \cA(\sigma) \right)  = \sup_{\varphi_{1}\perp\varphi_{2}}E_{e^\varepsilon}(\mathcal{A}(\varphi_{1}
)\Vert\mathcal{A}(\varphi_{2})) \leq \delta,
\end{equation}
where $\varphi_{1}$ and $\varphi_{2}$ are orthogonal pure
states.
This is implied by the proof of~\cite[Theorem~II.2]{hirche2023quantum}, and an alternative proof is given in~\cite[Appendix~A]{nuradha2024contraction}. 
 We also denote the set of all $(\varepsilon,\delta)$-QLDP mechanisms by \begin{equation} \label{eq:private-channels}
\mathcal{B}^{\varepsilon,\delta} \coloneqq \left\{ \cN \in \operatorname{CPTP}: \sup_{\rho,\sigma \in \cD(\cH)} E_{e^\varepsilon}\!\left( \cN(\rho) \Vert \cN(\sigma) \right) \leq \delta \right\}.
\end{equation}

\section{Privacy-Utility Tradeoff} \label{Sec:privacy_Utility}

In this section, we study tradeoffs associated with privacy and utility, while ensuring QLDP for quantum states. First, we find the optimal QLDP mechanism parameters by choosing a depolarizing channel as the private mechanism. Then, we discuss the utility attained by private mechanisms to obtain optimal utility and 
show that optimality is achieved by a depolarizing mechanism. 
\subsection{Optimal Local Differential Privacy Mechanisms}

\medskip
Consider the depolarizing channel acting on a $d$-dimensional system: $ \cA_{p}(\rho) = (1-p) \rho + p \tr(\rho)\frac{1}{d} I.$
Recall that a channel $\cA$ satisfies $(\varepsilon, \delta)$-quantum local differential privacy (QLDP) if and only if 
\begin{equation}
    \sup_{\rho, \sigma \in \cD(\cH)} E_{e^\varepsilon}\!\left( \cA(\rho) \Vert \cA(\sigma) \right) \leq \delta.
\end{equation}

In~\cite{hirche2023quantum,nuradha_QPP}, a sufficient condition on $p$ to achieve QLDP has been derived. Now, we show that those noise levels are not only sufficient, but also necessary to achieve the level of privacy demanded by QLDP. Note that the above claim is proven for $\delta=0$ case in~\cite{guan2024optimal} via different proof arguments.

\begin{proposition}[Optimal $p$ for QLDP] \label{prop:optimal_P_QLDP_Dep}
     Let $\varepsilon \geq 0$ and $\delta \in [0,1]$. Then, we have that 
    \begin{equation}
        \cA_p \textnormal{ is } (\varepsilon, \delta)\textnormal{-QLDP} \iff p \geq \frac{d(1-\delta)}{e^\varepsilon +d-1}.
    \end{equation}
In particular, $p^* = \frac{d(1-\delta)}{e^\varepsilon +d-1}$ is the optimal noise parameter for depolarizing channel 
$\cA_p$ to achieve $(\varepsilon, \delta)$-QLDP.
\end{proposition}
The proof utilizes~\eqref{eq:equivalence_pure} together with the fact that $\varphi_1$, $\varphi_2$ and $I-\varphi_1-\varphi_2$ are orthogonal projections for any pair of orthogonal pure states $(\varphi_1,\varphi_2)$, and is given in Appendix~\ref{App:optimal_p_QLDP}.

\subsection{Optimal Privacy-Utility of QLDP Mechanisms}
Here, we analyse the optimal utility achieved by $(\varepsilon, \delta)$-QLDP mechanisms. In~\cite{guan2024optimal}, this question is partially answered by showing that a depolarizing channel achieves some optimal utility measures (introduced next based on trace distance and fidelity) when optimized over all unital channels that satisfy $(\varepsilon,0)$-QLDP. The main contribution of our work is to obtain optimal utility for all $\varepsilon \geq 0$ and $\delta \in [0,1]$ under all $(\varepsilon, \delta)$-mechanisms having the same input and output dimensions; in fact, in this general setting, we show that a depolarizing channel with the right parameter choice achieves the optimal utility for certain unitarily invariant utility metrics defined next.

First, we recall two utility metrics of interest that were also studied in~\cite{guan2024optimal}, considering unital channels.
Here we define them for arbitrary, not necessarily unital, quantum channels mapping between isomorphic systems, i.e., $\cN\colon \cL(\cH_A) \to \cL(\cH_B)$ such that $d_A=d_B$. 
We will also keep this assumption in the following, unless specified otherwise.

\begin{definition} [Utilities] \label{def:fidelity_utility}
The fidelity utility of a channel $\cN$ is $F(\cN) \coloneqq  \min_{\rho \in \cD(\cH)} F\big(\cN(\rho),\rho\big),$
where $F$ is the fidelity defined in~\eqref{eq:fidelity}.
 The trace-distance and anti-trace-distance utility are defined as 
    $ T(\cN) \coloneqq  \max_{\rho \in \cD(\cH)} T(\cN(\rho), \rho)$ and $\widehat T(\cN) \coloneqq  1 - T(\cN),$ respectively.
\end{definition}

 In the following discussion we use the notation $\cN_U \coloneqq \cU^\dag \circ \cN \circ \cU = U^\dagger \cN(U\cdot U^\dagger)U$, where $\cU=U\cdot U^\dagger$ is a unitary channel. Also, recall the twirled channel $\cN_G$ with respect to the Haar measure on a compact group $G$ and two unitary representations of $G$ at the input and output of $\cN$, as given in \Cref{def:twirled_channel}.

\begin{proposition}
    \label{prop:N_u_N_G_private}
Let $\cN$ be a quantum channel. Then $\cN$ is $(\varepsilon,\delta)$-QLDP if and only if $\cN_U$ is.
Furthermore, if $\cN$ is $(\varepsilon,\delta)$-QLDP, then so is $\cN_G$.
\end{proposition} 
The proof of \Cref{prop:N_u_N_G_private} uses the unitary invariance and the joint convexity of hockey-stick divergence. See~\Cref{App:prop_N_U_N_G_private} for the complete proof. 

Also, note that the channels $\cN_U$ and $\cN_G$ are practically relevant: these channels occur when unitary channels act together with private (noisy) channels, and the twirled channel $\cN_G$ corresponds to the channel when these unitary channels are randomly chosen. 

\begin{lemma}
\label{lem:twirl-fidelity-qdps}
Let $\cN$ be a quantum channel and consider the twirled channel $\cN_G = \int_G \mathrm{d}\mu(g)\, \cN_{U_g}$ with respect to a fixed unitary representation $g\mapsto U_g$. Then, $  F(\cN_G) \ge F(\cN).$
Furthermore, we also have $T(\cN) \geq T(\cN_{G})$.
\end{lemma}
The proof of \Cref{lem:twirl-fidelity-qdps} follows by the unitary invariance of fidelity and trace distance together with the concavity and convexity of fidelity and trace distance, respectively. See~\cref{app:Lemma_twirl_trace} for the complete proof.

Note that the proof of Lemma \ref{lem:twirl-fidelity-qdps} shows that a similar result holds even if we replace the utility based on fidelity or trace distance by a function that satisfies unitary invariance and concavity (resp.~convexity).  

{ Next, we study the optimal utility achieved by private mechanisms satisfying $(\varepsilon,\delta)$-QLDP that have the same input and output dimensions (so that $F(\cN)$ and $T(\cN)$ are well defined) denoted by }
\begin{equation} \label{eq:epsilon_Delta_private_Channel}
\mathcal{B}^{\varepsilon,\delta}_d \coloneqq \left\{ \cN : \sup_{\rho,\sigma \in \cD(\cH)} E_{e^\varepsilon}\!\left( \cN(\rho) \Vert \cN(\sigma) \right) \leq \delta, \  \cN : \cL(\cH_A) \to \cL(\cH_B) \textnormal{ s.t. } d_A=d_B=d \right\} . 
\end{equation}
\begin{theorem}[Optimal Utility of QLDP Mechanisms]
  \label{thm:optimal-utility}
  Let $\varepsilon >0$ and $\delta \in [0,1]$. We have that
\begin{align}
    \max_{\cN \in \cB^{\varepsilon, \delta}_d} F(\cN) &=   \frac{e^\varepsilon +\delta (d-1)}{e^\varepsilon +d-1} \\ 
   \min_{\cN \in \cB^{\varepsilon, \delta}_d} T(\cN) &=  \frac{(d-1)(1-\delta)}{e^\varepsilon +d-1}, 
\end{align}
where $\cB_d^{\varepsilon,\delta}$ is defined in~\eqref{eq:epsilon_Delta_private_Channel}.
In particular, 
\begin{equation}
    \max_{\cN \in \cB^{\varepsilon, \delta}_d} F(\cN)=  \max_{\cN \in \cB^{\varepsilon, \delta}_d} \widehat{T}(\cN).
\end{equation}
\end{theorem}

\begin{proof} 
We first observe that for a depolarizing channel $\cA_p$, we have $F(\cA_p) = 1-\frac{p(d-1)}{d} 
$. 
Recall, by \Cref{prop:optimal_P_QLDP_Dep}, we have that $\cA_p \in \mathcal{B}^{\varepsilon,\delta}_d$ when $ p\geq \frac{d(1-\delta)}{e^\varepsilon + d-1}\eqqcolon p^*$. Also, note that $F(\cA_p)$ is maximized 
when $p=p*$ for $\cA_p \in \mathcal{B}^{\varepsilon,\delta}_d$. Then, by choosing $\cN= \cA_{p*}$, we have $\max_{\cN\in\mathcal{B}^{\varepsilon,\delta}_d} F(\cN) \geq 1-p^*(1-1/d). $ To obtain the matching upper bound, we recall that for a channel $\cM \in \mathcal{B}^{\varepsilon,\delta}_d$ achieving the maximum, the twirled channel $\cM_{\mathbb{U}(d)}$ is a depolarizing channel $\cA_q$, which also belongs to $\mathcal{B}^{\varepsilon,\delta}_d$ by \Cref{prop:N_u_N_G_private}. Using \Cref{lem:twirl-fidelity-qdps}, we have $F(\cM) \leq F(\cM_{\mathbb{U}(d)}) \leq 1-p^*(1-1/d)$, leading to the desired matching upper bound. See~\Cref{app:proof_optimal_utility_generic_d} for the complete proof.    
\end{proof}

\bigskip

Previously, we discussed the optimal utility of private channels that have the same input and output dimensions. However, one could apply a private channel that has different input-output dimensions as well. In that case, to compare its closeness to the original state, for the utility metric to make sense, we have to allow post-processing of the output to bring it back to the input dimension. In fact, this is a quite natural setting in private data-processing, since the adversaries or the data observers may apply a post-processing mechanism to optimize their utility.  With that in mind, we consider arbitrary channels $\cN\colon A\to B$ (with $d_A$ and $d_B$ potentially different) and define the following alternative utility function:
\begin{equation}
    \tilde{F}(\cN)= \max_{\cR_{B \to A}} \min_{\rho \in \cD(\cH_A)} F(\cR \circ \cN(\rho), \rho).
\end{equation}
where we optimize over post-processing channels $\cR_{B \to A}$. Note that 
$\cR \circ \cN$ is a channel with equal input and output dimensions.
 Therefore, \Cref{thm:optimal-utility} implies that 
\begin{equation}
    \tilde{F}(\cN) \leq   \max_{\cM \in \cB^{\varepsilon, \delta}_d} F(\cM) =   \frac{e^\varepsilon +\delta (d-1)}{e^\varepsilon +d-1}.
\end{equation}
Also, when optimizing over $\cN$ satisfying $(\varepsilon, \delta)$-QLDP, we get 
\begin{equation}
   \max_{\cN \in \cB^{\varepsilon,\delta}} \tilde{F}(\cN) \leq   \frac{e^\varepsilon +\delta (d-1)}{e^\varepsilon +d-1},
\end{equation}
where $\mathcal{B}^{\varepsilon,\delta}$ is defined in~\eqref{eq:private-channels}.
In fact, this upper bound is achievable by choosing $\cN$
to be a depolarizing channel with $p= d(1-\delta)/(e^\varepsilon +d-1)$
and $\cR$ as the identity channel, showing that even in this operational post-processing setting, the depolarizing channel achieves the optimal utility. Thus, we have 
\begin{equation} \label{eq:utility_optimized_post_processing}
 \max_{\cN \in \cB^{\varepsilon,\delta}}  \tilde{F}(\cN)=  \max_{\cN \in \cB^{\varepsilon,\delta}}  \max_{\cR_{B \to A}} \min_{\rho \in \cD(\cH_A)} F(\cR \circ \cN(\rho), \rho) =     \frac{e^\varepsilon +\delta (d-1)}{e^\varepsilon +d-1}.  
\end{equation}

\bigskip
\noindent \textbf{Privacy-Utility Tradeoffs:}
 In \Cref{fig:eps-optF-two-plots}, we plot the optimal utility that we obtained analytically in~\Cref{thm:optimal-utility} and~\eqref{eq:utility_optimized_post_processing} with respect to privacy parameters $(\varepsilon,\delta)$ and the system dimension $d \geq 2$. We observe that optimal utility grows when $\varepsilon$ increases, showcasing a privacy-utility tradeoff (recall that an increase in $\varepsilon$ corresponds to a decrease in the privacy level). Also, from Figure (a), we see that for small $\varepsilon$ (high privacy regime), it would be possible to improve the utility if we are willing to tune the parameter $\delta$. From (b), we see that the utility drops when the system size increases, highlighting that the optimal private mechanisms when $d$ is high reduce the optimal utility, which is dominant in the high privacy regime.  
     \begin{figure}[t]
  \centering
  \begin{minipage}[t]{0.5\textwidth}
    \centering
    \includegraphics[width=\linewidth]{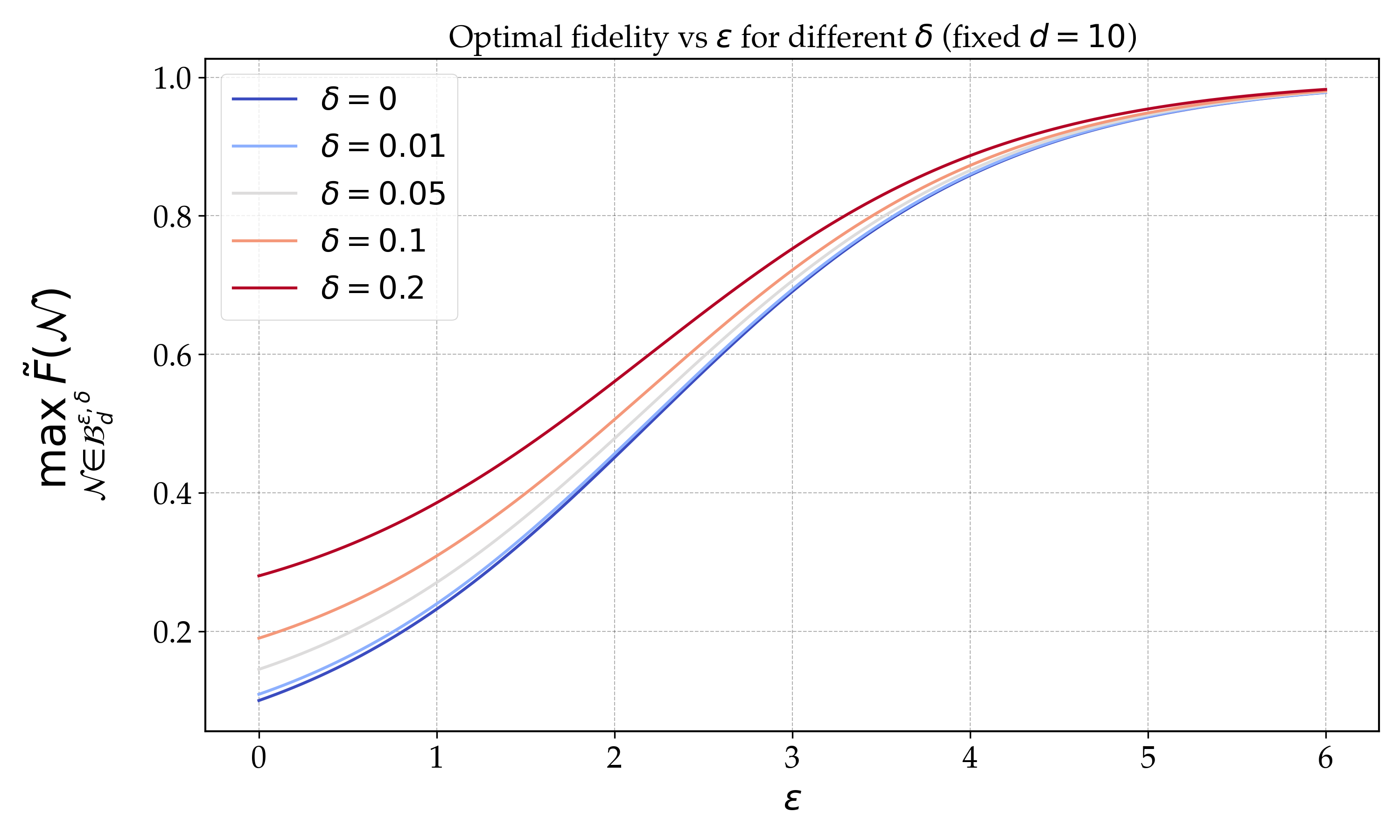}
    (a)
  \end{minipage}\hfill
  \begin{minipage}[t]{0.5\textwidth}
    \centering
\includegraphics[width=\linewidth]{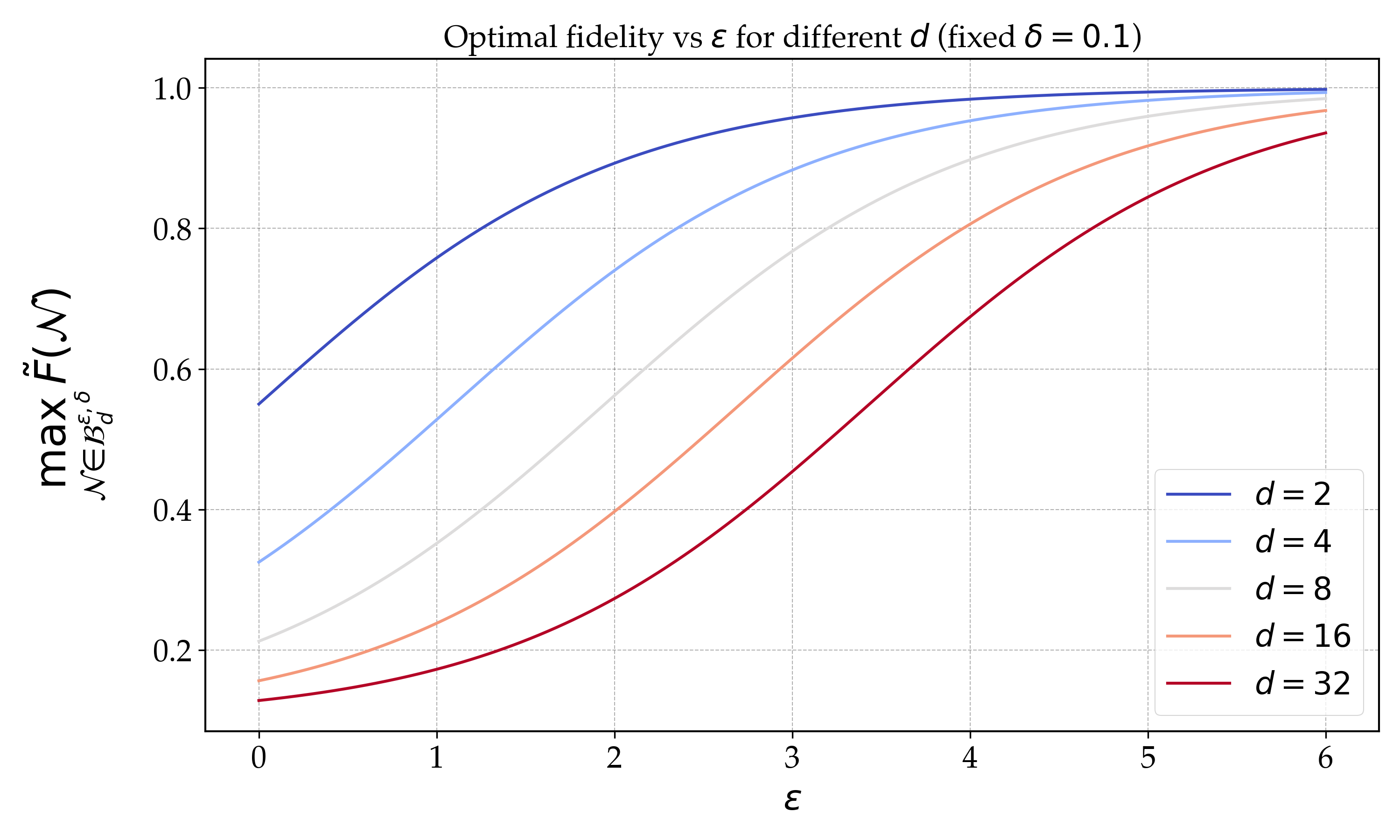}
    (b)
  \end{minipage}
  \caption{Optimal utility under $(\varepsilon,\delta)$-QLDP (\Cref{thm:optimal-utility} and~\eqref{eq:utility_optimized_post_processing}). 
    (a) $\varepsilon$ vs.\ optimal fidelity for varying $\delta$ (fixed $d=10$). 
    (b) $\varepsilon$ vs.\ optimal fidelity for varying $d$ (fixed $\delta=0.1$).}
  \label{fig:eps-optF-two-plots}
\end{figure}

\section{Privately Estimating Expectation of Observables} \label{Sec:private_observable_estimation}

In the previous section, we studied privacy-utility tradeoffs with generic utility measures, which are designed for the one-shot setting in which utility is measured with respect to one copy of the privatized state. 
However, these generic utility measures are not always the best measures in specific use cases of the privatized data, and when we have access to finitely many copies of the privatized state. This is, for example, the case when estimating a property of a quantum state using $n$ copies of the state. Here, we study the \emph{private} version of this property estimation task: we aim to learn a given property of a state with sufficient accuracy, but without full knowledge of the state. 

{Estimating the expectation of an observable is a fundamental question in quantum learning theory, which studies how efficiently properties of quantum states, such as purity, entanglement, or fidelity to a target state, can be estimated. In this work, we focus on how to perform this estimation task privately. In particular, we consider the setting where we need to estimate the expectation of a given observable using $n$ samples of private data without inferring the original state.} We now describe this task in more detail. Instead of getting $n$ copies of a state $\rho$, we get $n$ copies of the privatized states $\cA(\rho)$. This way, we make sure that our original inputs are private. With that, we ask the question how well one can estimate $\Tr[O\rho]$, with $O$ being an observable, having only access to the privatized samples $\cA(\rho)$. Throughout, we use the assumption that we know the full mathematical description of the observable at the estimation phase, but we do not know the mathematical description of the states, only having access to the samples of the states. 

In terms of this task, we need to estimate the observable to a sufficient accuracy with high probability with the minimum number of samples of privatized data. First, let us define our accuracy criterion and the optimal sample complexity in this task:
Let $\hat{E}(\cA, \rho,O,n)$ be the estimate for $\Tr[O \rho]$ obtained by a procedure $\cP_O$ that takes $\left(\cA(\rho)\right)^{\otimes n}$ as input such that the following condition is satisfied: $\left| \Tr[O\rho]-\hat{E}(\cA, \rho,O,n)\right| \leq \beta$
with probability at least $1-\eta$, for $\beta>0$ and $\eta \in (0,1)$.
We define the optimal sample complexity for estimating the expectation value while ensuring $\varepsilon$-QLDP as follows:
\begin{equation}\label{eq:learning_exp_SC_def}
    n^*(O, \beta, \eta, \varepsilon,\delta) \coloneqq \inf_{\cA \in \cB^{\varepsilon,\delta}, \cP_O}\left\{n \in \mathbb{N}:  \left| \Tr[O\rho]-\hat{E}(\cA, \rho,O,n)\right| \leq \beta \  \textnormal{w.p.} \ 1-\eta \ \forall \rho \in \cD \right\},
\end{equation}
where the optimization is over protocols $\cP_O$ and the private channels $\cA \in \cB^{\varepsilon,\delta}$ defined in~\eqref{eq:private-channels}.
When $\delta=0$, we use the shorthand $ n^*(O, \beta, \eta, \varepsilon)  \equiv  n^*(O, \beta, \eta, \varepsilon,\delta=0) $.

{{ There are two operational settings for the task of privately estimating the expectation of an observable: 
In the first setting, the \emph{same} party generates the $n$ copies of the state from a quantum device sequentially, with the goal of allowing the analyzer to reliably estimate $\Tr[O\rho]$ without revealing the true identity of the state $\rho$. 
In the second setting, $n$ distinct parties each hold a copy of the state $\rho$.
After receiving the privatized copies from the distinct parties, the central analyser is able to learn a specific property of this state, but not the actual state itself.}}

There are two main questions related to characterizing optimal sample complexities: First, we aim to find $\cA \in \cB^{\varepsilon}$ and a protocol $\cP_O$ to obtain an achievable upper bound on the optimal sample-complexity. Second, we aim to obtain a non-trivial lower bound on the minimum number of samples needed for any $\cA\in \cB^{\varepsilon,\delta}$ and $\cP_O$. In fact, the goal is to obtain matching upper and lower bounds for as many parameters as possible with the least number of assumptions. This question has been studied in~\cite[Proposition~33]{gallage2025theory}, however, those results only hold for a certain class of observables with $0 \leq O \leq I$, and a certain distinguishability criterion on the observable and states of interest with respect to the accuracy level. Moreover, the sample complexity therein is only tight in terms of the privacy parameter, but not in terms of the accuracy parameters.{Ref.~\cite{angrisani_localModel25} provides a way to generate achievability protocols in the setting $\delta=0$ through applying private POVMs, which is a special class of private channels.}

\subsection{Lower Bound}

To obtain a lower bound on the estimation sample complexity, 
we make use of the lower bounds on private quantum hypothesis testing (QHT) studied in~\cite{nuradha2024contraction,Christoph2024sample}, which is also given in~\Cref{App:Private_QHT} for completeness.

\begin{proposition} \label{prop:LB_1}
Let $ 0<\beta \leq (\lambda_{\max}(O) -\lambda_{\min}(O))/4  $, where $O$ is the observable of interest with $\lambda_{\max}(O) > \lambda_{\min}(O)$, and $ \eta \in (0,1/4) $ and $\varepsilon >0$.
Then, we have that 
\begin{equation}
      n^*(O, \beta, \eta, \varepsilon)  \geq \frac{\ln\! \left( \frac{1}{4 \eta (1-\eta)}\right) e^\varepsilon  (\lambda_{\max}(O) -\lambda_{\min}(O))^2}{ 32(e^\varepsilon -1)^2  \beta^2 }.
\end{equation}
Furthermore for $\varepsilon \in (0,1)$, we have 
\begin{equation}
     n^*(O, \beta, \eta, \varepsilon)  \geq  \Omega\!\left( \frac{(\lambda_{\max}(O) -\lambda_{\min}(O))^2 \ln(1/\eta)}{ \varepsilon^2 \beta^2} \right).
\end{equation}    
\end{proposition}
    To prove the lower bound in \Cref{prop:LB_1}, we first show that if there exists a protocol that estimates $\Tr[O\rho]$ up to the given accuracy conditions, one can use that protocol to devise a hypothesis testing protocol that discriminates between two selected states. Then, for this setting, the lower bounds on private QHT apply (\cref{thm:bounds_sample_C_private}) and yield a lower bound to the private observable estimation task. See~\cref{app:estimation_LB_proof} for a detailed proof.
    
We note here that we always have $\lambda_{\max}(O) -\lambda_{\min}(O) \geq 0$ for any observable $O$, and if $\lambda_{\max}(O) =\lambda_{\min}(O)$ then $O=cI$ with $c = \lambda_{\min}(O) = \lambda_{\max}(O) \in \mathbb{R}$. In the latter case we can perfectly estimate $\Tr[O\rho]$ without even having access to the (fully private) states, since it evaluates to $\Tr[O\rho] =c$. Thus, we only consider the case $\lambda_{\max}(O) -\lambda_{\min}(O) > 0$ in the sequel.

\subsection{Upper Bounds}
Next, we devise private mechanisms and protocols for the task of estimating $\Tr[O\rho]$ with access to $n$ copies of $\cA(\rho)$, where $\cA \in \cB^{\varepsilon,\delta}$.

\subsubsection{Using Pauli Measurements}
Consider the representation of the observable $O$ in terms of Pauli operators as $O= \sum_P {\alpha_P}P $ with $P \in \cP_m$ being a Pauli operator on $m$-qubits (see~\eqref{eq:pauli_operators}) and $\alpha_P \in \mathbb{R}$, and denote $S=\sum_P |\alpha_P|$.
We define a probability distribution $p$ over the set of Paulis $\cP_m$ with $p(P)=|\alpha_P|/S$.

Let us then consider the following channel denoted by $\cA_q^O$: 
\begin{equation} \label{eq:channel_private_for_observable}
    \cA_q^O(\rho)= \sum\nolimits_{P} p(P) \ \cA_q \!\left( \cM_P (\rho)\right) \otimes |P \rangle\!\langle P|,
\end{equation}
where $\cA_q$ is a qubit depolarizing channel, the $P$ register has dimension $4^m$ (with orthonormal basis $\{ |P \rangle \}$), and $   \cM_P(\rho) = \Tr\!\left[ \left(\frac{I+P}{2}\right) \rho \right] | 0\rangle \!\langle 0| + \Tr\!\left[ \left(\frac{I-P}{2}\right) \rho \right] | 1\rangle \!\langle 1|.$
  The channel $\cA_q^O$ is a composition of the following protocol, leading to the output of the channel to be classical with two registers (the first one with either $0$ or $1$, and the second one indicating the chosen Pauli (P)):
\begin{enumerate}
  \item Sample a Pauli $P$ from $\mathcal{P}_m$ according to $p(P)=|\alpha_P|/S$ 
  and store it in a separate register.
  \item Apply the measurement channel $\cM_P$ corresponding to $P$ to $\rho$. 
  \item Apply a depolarizing channel with parameter $q=2(1-\delta)/(1+e^\varepsilon) $ to the output of $\cM_P(\rho)$.
\end{enumerate}

\begin{lemma} \label{lem:A_O^q_private}
The channel $\cA_q^O$ in~\eqref{eq:channel_private_for_observable} satisfies $(\varepsilon,\delta)$-QLDP. 
\end{lemma}
See~\cref{app:lemma_A_O^q_private_proof} for the proof of \Cref{lem:A_O^q_private}. 
Now, given $n$ copies of the outputs of the above stated channel $\cA_q^O$, we provide a method to estimate $\Tr[O\rho]$ with a bound on the number of samples $n$ required to achieve the desired accuracy criterion.
\begin{proposition} \label{prop:upper_1}
Let $\varepsilon >0, \delta \in [0,1]$, $\beta >0$, and $\eta \in (0,1) $. For an observable $O= \sum_{P} \alpha_P P$ with $S \coloneqq \sum_P |\alpha_P|$, there exists a protocol that requires   \begin{equation}
         n  \geq \frac{2 \ S^2 (e^\varepsilon +1)^2}{\beta^2 (e^\varepsilon-1+2\delta)^2} \ln\!\left(\frac{2}{\eta}\right)
     \end{equation} copies of $\cA_q^O$ in~\eqref{eq:channel_private_for_observable}
    to achieve $\left|\Tr[O\rho]-\hat{E}(\cA, \rho,O,n)\right| \leq \beta \  \textnormal{w.p.} \ 1-\eta$. Moreover, for $\varepsilon \in (0,1)$ and $\delta=0$, we have 
\begin{equation}
          n^*(O, \beta, \eta, \varepsilon)   \leq \left\lceil \frac{ 32 S^2 \ln(2/\eta)}{ \varepsilon^2 \beta^2} \right \rceil.
    \end{equation}
\end{proposition}
The proof idea of \Cref{prop:upper_1} is as follows:
Assume that we are given $n$ i.i.d.\ copies of $\mathcal A_q^O(\rho)$ in~\eqref{eq:channel_private_for_observable}.
For each sample $i\in\{1,\dots,n\}$, let the observed pair be $(Y_i, P_i)$, where $Y_i\in\{0,1\}$ is the computational-basis bit obtained after applying the channels $\cA_q \circ \cM_{P_i}$ and $P_i\sim p(\cdot)$ is the Pauli label. Using the pair $(Y_i, P_i)$, we define the random variables $X_i \coloneqq  (-1)^{Y_i}\in\{+1,-1\} $ and $Z_i \coloneqq \frac{S}{1-q}\,\operatorname{sgn}(\alpha_{P_i}) X_i.$
We then choose an estimator $ \widehat{E} \coloneqq \frac{1}{n}\sum_{i=1}^n Z_i$ and show that it is an unbiased estimator (i.e., $\mathbb{E}[\widehat{E}]= \Tr[O\rho]$). Finally, the sample complexity in the statement of the proposition follows from Hoeffding's inequality. We refer to~\cref{app:prop_upper_bound_estimate_1} for the complete proof.

\subsubsection{Private Classical Shadows}
Classical shadows have been identified as an approach to estimate properties of quantum states, possibly efficiently in some settings \cite{huang2020predicting,huang2022learning}. Also, classical shadows have been studied in \cite{Koh2022classicalshadows} when noise channels act along with random unitary circuits. In this section, we discuss how we can utilize carefully designed noise levels to achieve a required privacy level while providing useful estimates of classical shadows in estimation tasks. To the best of our knowledge, this is the first work discussing private classical shadows, which may be a useful tool in private information processing with quantum systems.

First, we present the approach to obtain a classical shadow from one copy of the state $\rho$ (for an $m$-qubit state with $d=2^m$), by fixing the random unitary gates to be Clifford and the noise channel to be a depolarizing channel $\rho$: (i) Given an input state $\rho$, we sample a unitary $U$ from the set of Clifford unitaries $\mathcal{C}_m$ (see~\eqref{eq:clifford_group}) and apply the selected unitary to the input state, giving $U\rho U^\dag$ as the output; (ii) Apply a depolarizing channel with flipping parameter $\hat{p} \in [0,1]$ (note that one could use any private channel with the right parameters); (iii) Perform a computational basis measurement to get outcomes $|b\rangle\!\langle b|$ with $ b\in \{0,1\}^m$; (iv) Apply the unitary channel $U^\dag (\cdot) U$ to the classical outcomes to obtain $U^\dag |b\rangle\!\langle b| U$. 
    The composite channel after this step is as follows:
{ \begin{equation}\label{eq:composite_channel_clas_shad}
        \cL_{\hat{p},d}(\rho) = \mathbb{E}_{U\sim \mathcal{C}_m}\! \left[ \sum\nolimits_{b \in\{0,1\}^m} \langle b| \cA_{\hat{p}}(U \rho U^\dag )|b \rangle \ \ket{U,b}\bra{U,b}\  \right] ;
        \end{equation}
 where $|U,b\rangle$  represents the classical description of the random Clifford unitary $U$ and classical outcome $b$ after the computational basis measurement; }
(v) Evaluate $\hat{\rho}$, where $\hat{\rho} \coloneqq \frac{d+1} {1-\hat{p}}  U^\dag |b\rangle\!\langle b| U- \left( 1- \frac{d+1} {1-\hat{p}} \right) \frac{I}{d},$
    and send it to the central party to evaluate the expectation over the given observable $O$.

This procedure is done for all $N$ copies of $\rho$ (possibly within local systems), and the $\hat{\rho}$ is sent for the evaluation. With access to $N$ copies of $\hat{\rho_i}$ for $i=\{1,\ldots, N\}$, we follow the following steps to obtain an estimate for $\Tr[O \rho]$. Let $N=k \ell$ with $k,\ell \in \mathbb{N}$. 
\begin{enumerate}
    \item First we obtain $\hat{\rho}_{(k)} =\frac{1}{\ell} \sum_{i= (k-1) \ell +1}^{ k \ell } \hat{\rho}_i$ for $k=\{1, \ldots, N/\ell\}$. 
    \item Then release $\hat{o}$ as the estimate for $\Tr[O\rho]$ , where 
    \begin{equation}
        \hat{o} =\operatorname{median}\!\left\{ \Tr[O \hat{\rho}_{(1)}], \ldots,  \Tr[O \hat{\rho}_{(N/\ell)}]\right\}.
    \end{equation}
\end{enumerate}
{
  \begin{proposition}[Sample Complexity of Private Classical Shadow Approach] \label{prop:upper_2}
Let $\varepsilon >0$, $\delta \in [0,1]$, $\beta >0$ and $\eta \in (0,1)$. To achieve  $\Pr \left( \left| \Tr[O\rho]-\hat{o} \right| \leq \beta\right) \geq 1- \eta$, it is sufficient to have $N$ copies of private classical shadows where
\begin{equation}
     N \geq \frac{204  \Tr[O^2] }{\beta^2 } \left( \frac{e^\varepsilon +d-1} {e^\varepsilon -1+d\delta} \right)^2  \ln\!\left(\frac{2}{\eta}\right).
\end{equation}
This is achieved by choosing $ \hat{p} =\frac{d(1-\delta)}{e^\varepsilon+d-1}$ in~\eqref{eq:composite_channel_clas_shad}. Furthermore for $\delta=0$
and $\varepsilon \leq \ln(2)$, we have 
\begin{equation}
     n^*(O, \beta, \eta, \varepsilon)   \leq \left\lceil \frac{3264  \Tr[O^2] (d-1)^2 }{\beta^2 \varepsilon^2} \ln\!\left(\frac{2}{\eta}\right) \right \rceil.
\end{equation} 
\end{proposition}
\begin{proof}
    By~\cite[Corollary~4.11]{Koh2022classicalshadows} (with $f=1-\hat{p}$), this procedure ensures that 
\begin{equation}
\Pr \left( \left| \Tr[O\rho]-\hat{E}(\cA, \rho,O,n)\right| \leq \beta\right) \geq 1- \eta
\end{equation}
for $\eta \in (0,1)$ and $\beta>0$ 
if 
\begin{equation} \label{eq:SC_Noisy_Classi_Shado_1}
    N \geq \frac{204  \Tr[O^2] }{\beta^2 (1-\hat{p})^2} \ln\!\left(\frac{2}{\eta}\right).
\end{equation}
Now, we need to ensure that the above mechanism executed locally to obtain $\hat{\rho}$ satisfies $(\varepsilon,\delta)$-QLDP in order to be considered a private mechanism. 
One way to accomplish this is by choosing $\hat{p}$ of the depolarizing channel as 
$\hat{p}= d(1-\delta)/(e^\varepsilon +d-1)$ by using \Cref{prop:optimal_P_QLDP_Dep} (recall that this mechanism achieves optimal utility in~\cref{thm:optimal-utility}) and recalling the data processing property of QLDP mechanisms. In particular, we use the fact that $\cL_{\hat{p},d} = \cK \circ \cA_{\hat{p}}$ where 
        $\cK(\omega) \coloneqq \mathbb{E}_{U\sim \mathcal{C}_m}\! \left[ \sum\nolimits_{b \in\{0,1\}^m} \langle b| (U \omega U^\dag )|b \rangle \ \ket{U,b}\bra{U,b} \right]$
along with $U \cA_{\hat{p}}(\rho) U^\dag =\cA_{\hat{p}}(U \rho U^\dag )$.

For the setting $\varepsilon \leq 1$, consider the following:
  \begin{align}
        \left( \frac{e^\varepsilon -1+d}{e^\varepsilon-1}\right)^2  & =\left( \frac{e^\varepsilon +1}{e^\varepsilon-1} + \frac{d-2}{e^\varepsilon-1}\right)^2  \\
        & \leq \left( \frac{e^\varepsilon +1}{e^\varepsilon-1} + \frac{(d-2)(e^\varepsilon +1)}{e^\varepsilon-1}\right)^2  \\
        &= (d-1)^2 \left( \frac{e^\varepsilon +1}{e^\varepsilon-1}\right)^2 \\
        &\leq \frac{16 (d-1)^2}{\varepsilon^2},
    \end{align}
where the last inequality follows since $\varepsilon \leq 1$ and~\eqref{eq:epa_one_less_upper}. Plugging this upper bound provides the desired inequality on the optimal sample complexity.

\end{proof}}

{

The dependence of $d$ in~\Cref{prop:upper_2} limits the usage of private classical shadows for the task of estimating $\Tr[O\rho]$ privately for increasing dimension. However, it could be the case that the current bound can be improved by using different proof techniques or a relaxed privacy criterion. Thus, we leave further exploration of private classical shadows for future work. 

\begin{remark} [Privacy of the Shadow Channel]
    In standard classical shadows, the shadow channel refers to the channel where classical outcomes are discarded:
\begin{equation}
     \cS_{\hat{p},d}(\rho) = \mathbb{E}_{U\sim \mathcal{C}_m}\! \left[ \sum\nolimits_{b \in\{0,1\}^m} \langle b| \cA_{\hat{p}}(U \rho U^\dag )|b \rangle \ U^\dag |b\rangle\!\langle b| U \right].
\end{equation}
It is also known that $\cS_{\hat{p},d}$ is equivalent to a depolarizing channel (see \cite[Claim~4.10]{Koh2022classicalshadows}) with flipping parameter $q=1-(1-p)/(d+1)$, which can be used to show that a certain value of $\hat{p}$ is sufficient to achieve $(\varepsilon,\delta)$-QLDP. 
In particular, the shadow channel $\cS_{\hat{p},d}$ is $(\varepsilon,\delta)$-QLDP if 
    \begin{equation}
        \hat{p} =  1-  \min \left\{1, \frac{(e^\varepsilon -1 +d \delta) (d+1)}{e^\varepsilon +d-1}\right\}.
    \end{equation}
See~\cref{app:classical_shadow_SC_improved} for a detailed proof.
If we have a protocol that releases a quantum state (prepare a state from classical outcomes $U,b$) as in $\cS_{\hat{p},d}$, we can strengthen the privacy level of that protocol. However, to estimate $\Tr[O\rho]$ with the privatized states through this channel requires a different protocol $\cP_O$ other than the standard classical shadow approach. This is because the output of the shadow channel is a quantum state while the studied protocol in this work takes classical descriptions $\hat{\rho}$ that can be evaluated with the classical outcomes $b$ and the Clifford unitary chosen $U$.
\end{remark}}

\subsection{Optimality}
We now derive a tight characterization of the sample complexity for $\varepsilon \leq 1$. 

\begin{proposition} \label{prop:optimal_SC_O}
    Let $O =\sum_P \alpha_P P$ be an observable with $\lambda_{\max}(O) > \lambda_{\min}(O)$. Let $ 0<\beta \leq (\lambda_{\max}(O) -\lambda_{\min}(O))/4$, $ \eta \in (0,1/4) $ and $\varepsilon \in (0,1)$. Then, we have that 
    \begin{equation}
       \Omega\!\left( \frac{(\lambda_{\max}(O) -\lambda_{\min}(O))^2 \ln(1/\eta)}{ \varepsilon^2 \beta^2} \right) \leq    n^*(O, \beta, \eta, \varepsilon)   \leq O\!\left( \frac{ S^2 \ln(1/\eta)}{ \varepsilon^2 \beta^2} \right),
    \end{equation}
    where $S \coloneqq \sum_P |\alpha_P|$.

When $O$ is fixed, we have that 
\begin{equation}
     n^*(O, \beta, \eta, \varepsilon)   = \Theta\!\left( \frac{\ln(1/\eta)}{ \varepsilon^2 \beta^2} \right).
\end{equation}
\end{proposition}
    The first statement of \Cref{prop:optimal_SC_O} follows from combining \Cref{prop:LB_1} and \Cref{prop:upper_1}. 
    The second statement follows immediately from the first statement when $O$ is fixed. Furthermore for $\varepsilon \leq \ln(2)$, from \Cref{prop:upper_2} together with \Cref{prop:optimal_SC_O}, we get 
\begin{equation}
    \Omega\!\left( \frac{(\lambda_{\max}(O) -\lambda_{\min}(O))^2 \ln(1/\eta)}{ \varepsilon^2 \beta^2} \right) \leq    n^*(O, \beta, \eta, \varepsilon)   \leq O\!\left( \frac{ \min \left\{ S^2, \Tr[O^2] {(d-1)^2} \right\} \ln(1/\eta)}{ \varepsilon^2 \beta^2} \right).
\end{equation}

\medskip
\noindent \textbf{Optimality with Constants:}
{While the gap between the upper and the lower bounds can be quite large for generic observables, tighter bounds are possible when more information about the observable $O$ is available.}
Specifically, we consider the following observables to see how sample complexity scales in these settings:

\begin{enumerate}
    \item \textit{Pauli observables:}  
    When $O=P$ is itself a Pauli observable, we have $\lambda_{\max}(O)- \lambda_{\min}(O)=2$, and $S=1$. 
    In this case, the private sample complexity in the first statement in \Cref{prop:optimal_SC_O} is tight. In contrast, the classical shadows approach may not achieve optimality in this setting for large $d$, since $\Tr[O^2]=d=2^m$.

\item  \textit{Measurement Operators:}  Let $O$ be a measurement operator satisfying $0 \leq O \leq I$. Also, assume that $O$ could be implemented as a measurement in its circuit form (with the measurement channel represented by $\cM_O$). Then, using $n$ copies of the privatized state $\cA_q \!\left( \cM_O (\rho)\right) $ with $q=2(1-\delta)/(e^\varepsilon+1)$, we may use arguments similar to those in the proof of \Cref{prop:upper_1} to show that $n \geq \frac{2 \  (e^\varepsilon +1)^2}{\beta^2 (e^\varepsilon-1+2\delta)^2} \ln\!\left(\frac{2}{\eta}\right) $  is sufficient to achieve the required accuracy guarantees. With that, for $\varepsilon \leq 1$ and $\delta=0$, this said approach is optimal up to constants together with the lower bound in \Cref{prop:LB_1} whenever there exists a constant $C >0$ such that $\lambda_{\max}(O) -\lambda_{\min}(O) \geq C$.

\end{enumerate}

\medskip
\noindent 
{
\begin{remark}[Knowledge about the Observable] We state a scenario where the upper and lower bounds can be tight even with the dependence on $O$.
To this end, assuming that we know the spectrum and the eigenbasis of $O$, one can achieve the scaling of $(\lambda_{\max}(O)-\lambda_{\min}(O))^2$ instead of $S^2$ or $\Tr[O^2]$. In particular, given the decomposition $O= \sum_i \lambda_i \Pi_i$ such that $ \sum_i \Pi_i=I$ and $\Pi_i \geq 0$, one can apply the POVM corresponding to the eigenbasis $\{ \Pi\}_i $, then record the outcomes $\lambda_i$ and do classical post-processing together with classical private mechanisms by following approaches similar to~\cite[Section~4.2.3]{duchi2018minimax}.
In this protocol, we first scale outcomes to be in between [0,1] with $x_i= (\lambda_i- \lambda_{\min})/(\lambda_{\max}- \lambda_{\min})$, then apply random rounding with those scaled values defining the probabilities of getting 0 and 1 with $x_i$ and $1-x_i$, respectively, and then apply a classical randomized response with $p=2(1-\delta)/(e^\varepsilon +1)$.  However, this approach requires knowledge of the eigenbasis of $O$ and an efficient implementation of the respective POVM, which may not be the case in relevant scenarios. In special cases where the observable of interest is diagonal in the computational basis, we can adapt the said approach, having optimal sample complexity, since computational basis measurements can be efficiently implemented. 
\end{remark}}

\medskip
\noindent \textbf{Communication Cost}: Communication cost refers to the resources required by local parties to communicate each copy of $\cA(\rho)$ to the central party. 
In the first approach (\Cref{prop:upper_1}), we need to transmit a classical bit and the randomly chosen Pauli operator, requiring $2m +1$ classical bits of communication. However, for the second approach with classical shadows (\Cref{prop:upper_2}), we need to send a matrix of size $d=2^m$, having $d^2$ complex entries. To this end, depending on the precision required, the communication cost can be much higher. A reconstructed shadow is a $d$ by $d$ matrix with $d=2^m$ and needs to be transmitted to the central party by the local party who has access to it. One way to send it is to send every entry of the matrix separately. Also, note that the reconstructed shadow matrix does not always satisfy the conditions of a quantum state, so we need all $d^2$ entries to be sent. In this approach, depending on the precision required, the communication cost can be much higher. Another way to transmit the reconstructed shadow is to send the randomly chosen Clifford unitary and classical outcomes after the computational basis measurement separately so that the central party can generate the reconstruction with those.  This process would need $m$ bits to send classical outcomes while an $m$-qubit Clifford unitary can be specified using at most $ \log_2 |\mathcal{C}_m| \le
2m^2+3m$ bits using an efficient algorithm \cite{koenigsmolin}. With that, a local party needs at most $2m^2+4m$ bits to transmit the required information to the central party.

On the other hand, the classical shadow approach may be more advantageous (i) when exact knowledge of the observable $O$ is not available (recall that the first approach assumes the complete description of $O$ in designing the privatization mechanism) or (ii) when we aim to estimate many properties using a set of observables. Furthermore, if one privatizes each copy of the state such that generic utility measures are optimized (\cref{thm:optimal-utility}), it is required to send the quantum state $\cA_{p}(\rho)$ to the central party. This way, we also incur a cost for \textit{quantum communication} in contrast to the previous approaches, showcasing an added benefit of selecting the right utility measure depending on the task.

{\medskip
\noindent \textbf{Estimating many Observables:} When estimating many observables at once, i.e., estimating $\Tr[O_i \rho]$ for $i\in \{1,\ldots, M\}$, with access to privatized states, we need to pay for more samples. By using the private classical shadow approach, we need to pay a cost scaling with $\log(M)$ \cite[Corollary 4.11]{Koh2022classicalshadows} while the measure and randomize approach with Pauli measurements requires the sample complexity to scale linearly with $M$ (by repeating the single observable protocol $M$ times).
\begin{proposition} \label{prop:manyObservable}
   Let $\varepsilon >0$, $\delta \in [0,1]$, $\beta >0$ and $\eta \in (0,1)$. To achieve  $\Pr \left( \left| \Tr[O_i\rho]-\hat{o}_i \right| \leq \beta\right) \geq 1- \eta$ for all $i \in\{1,\ldots,M\}$, it is sufficient to have access to $N$ copies of private classical shadows where
\begin{equation}
     N \geq \frac{204  \max_{i \in \{1,\ldots, M\}} \Tr[O_i^2] }{\beta^2 } { \left( \frac{e^\varepsilon +d-1} {e^\varepsilon -1+d\delta} \right)^2 }\ln\!\left(\frac{2M}{\eta}\right).
\end{equation} 
\end{proposition}
The proof of this proposition uses similar arguments as those in~\Cref{prop:upper_2} together with~\cite[Corollary 4.11]{Koh2022classicalshadows} for $M \geq 1$.
}
{With the improved scaling with respect to $M$, the private classical shadow approach is more suitable if we have to estimate the observable expectation for a large number of observables.} {However, note that when $d$ is large and the privacy guarantees are within the high privacy regime, the advantage from the improved scaling may diminish.}

{To illustrate the dependence on the number of target
observables and the system dimension, consider $m$-qubit systems with $d=2^m$, and let
$O_i=\lvert\psi_i\rangle\!\langle\psi_i\rvert$, $i=1,\ldots,M$, be
rank-one projectors so that
$\max_i \Tr[O_i^2]=1$. Fix
$$
(\epsilon,\delta)=(0.5,0), \qquad
\beta=0.2, \qquad \eta=0.05.
$$
The depolarizing parameter used in the private classical shadow
mechanism is $$\hat{p_d} = \frac{d}{e^{0.5}+d-1}.$$ In particular, for $d=4$, we have $\hat{p_d} \approx 0.8605$. \Cref{prop:manyObservable} then gives
the sample complexity bound
$$
N(d,M)
 =
 \left\lceil
 \frac{204}{0.2^2}
 \left(
 \frac{e^{0.5}+d-1}{e^{0.5}-1}
 \right)^2
 \ln(40M)
 \right\rceil.$$

The dependence of the sufficient sample bound on $M$ and $d$ is
illustrated in \Cref{fig:private-shadow-scaling}. In panel~(a), $M$
ranges over powers of $10$ from $1$ to $10^4$, while
$d\in\{2,4,8,16\}$. For each fixed dimension, increasing $M$ from $1$
to $10^4$ increases the sufficient sample bound by only a factor of
approximately $3.50$, reflecting its logarithmic dependence on $M$.
For example, when $d=4$, each tenfold increase in $M$ adds approximately
$603027$ samples to the sufficient bound.

Panel~(b) fixes $M=10^2$ and illustrates the dependence on the system
dimension. Increasing the dimension from $d=2$ to
$d=16$ increases the bound by a factor of
\[
\left(
\frac{e^{0.5}+15}{e^{0.5}+1}
\right)^2
\approx 39.51.
\]
More generally, for fixed $M$ and $\delta=0$, the expression is proportional to $(d+e^{0.5}-1)^2$. The comparison
with the leading term $Cd^2$ in panel~(b) therefore demonstrates the
quadratic dependence on $d$. Thus, with all other parameters fixed, the
sufficient sample bound grows logarithmically with the number of
observables and quadratically with the system dimension.

\begin{figure}[t]
    \centering
    \includegraphics[width=\linewidth]{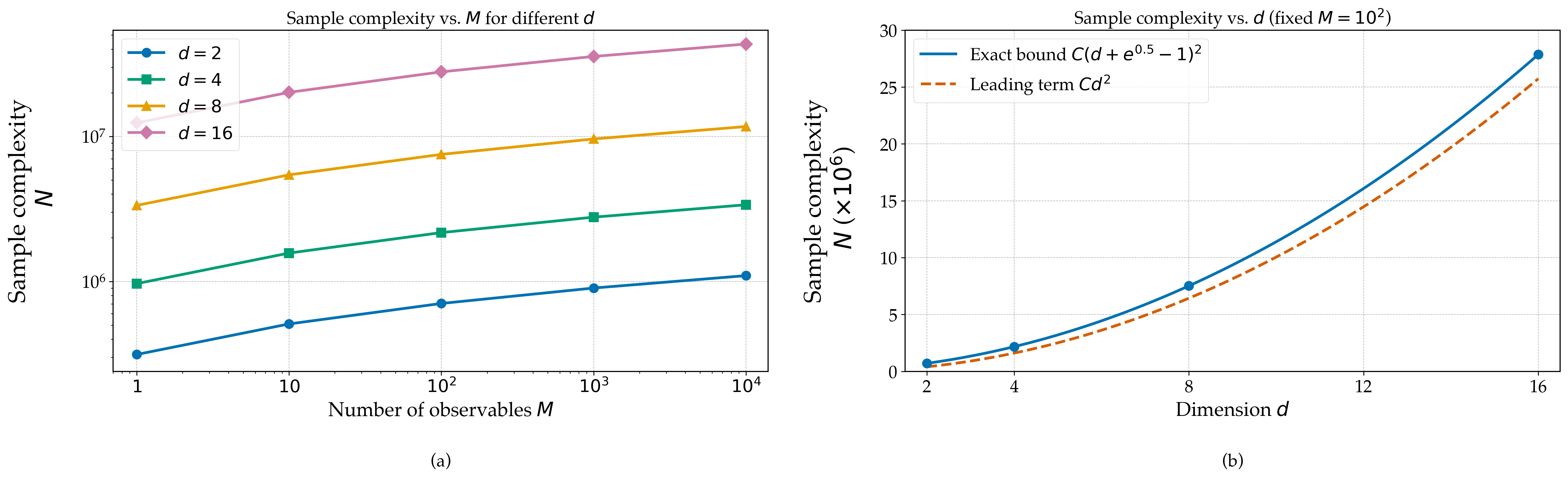}
    \caption{Sample complexity bound for private classical shadows from
\Cref{prop:manyObservable} for rank-one projectors, with
    $(\epsilon,\delta)=(0.5,0)$, $\beta=0.2$, and $\eta=0.05$.
    (a) Dependence of the sufficient sample complexity $N$ on the number
    of observables $M$ for several system dimensions $d$, illustrating
    the logarithmic dependence on $M$.
    (b) Dependence on the system dimension for fixed $M=10^2$.
    The solid curve is the bound
$C(d+e^{0.5}-1)^2$, while the dashed curve is its leading
    quadratic term $Cd^2$ for some constant $C$.}
\label{fig:private-shadow-scaling}
\end{figure}
}

\medskip
{We also provide a summary of the differences between the two approaches that we discussed for private observable expectation estimation in~\Cref{tab:summary_estimate_observable}, where one can use the approach that suits the given application setting and relevant constraints.}

  \begin{table}[h]
     {
        \centering
        \caption{Summary of the Comparison of Classical-Shadow and Pauli-Measurement Approaches}
        \begin{tabular}{|c|c|c|}
         \hline
             &  Classical-shadow approach & Pauli-measurement approach\\
             \hline
         Prior knowledge of $O$    &  No  & Yes \\
         \hline
          Communication cost &  
          \makecell[c]{$2^{2m}$ complex entries up to required precision,\\
or { $2m^2+4m$ bits using~\cite{koenigsmolin}}}
          &  $2m+1$ bits
         \\ 
         \hline
         Estimating $M >1$ observables  &  & \\ (sample complexity scaling for fixed $d$) & $\log(M)$ & $M$ \\
         \hline
         Scaling with dimension & through both $O$ and privacy parameters & only through $O$ \\
         \hline
        \end{tabular}

        \label{tab:summary_estimate_observable}}
    \end{table}
\section{Conclusion and Future Work} \label{Sec:Conclude}

In this work, we studied privacy-utility tradeoffs that exist in quantum information processing when privacy is quantified by $(\varepsilon,\delta)$-quantum local differential privacy. 
{In our work, we considered two main types of utility measures.
   \begin{enumerate}
       \item Generic utility measures (\Cref{def:fidelity_utility}): These are applicable when we do not have any information on the use case of the data. In particular, we optimize the similarity (fidelity) and closeness (trace distance) of the privatized state with respect to the original state while considering the worst-case input state. We showed that a depolarizing channel with parameter $p=d(1-\delta)/(e^\varepsilon +d-1)$ achieves the optimal utility for both fidelity and trace distance based utilities.
       \item Application-specific utilities~(eq.~\eqref{eq:learning_exp_SC_def}): We consider the application of estimating an expectation of an observable $O$ with respect to the input state $\rho$ (i.e, $\Tr[O\rho]$) up to an additive error with high probability with access to $n$ copies of the privatized state instead of the original state $\rho$. Here, we define the utility in terms of the optimal number of copies $n$  of the privatized state needed to achieve the required accuracy criterion by optimizing over all private channels and protocols. We propose two protocols based on a Pauli-measurement-based approach and a classical-shadow approach. By obtaining a lower bound with the assistance of private quantum hypothesis testing, we show that the optimal utility can be achieved  when $O$ is a Pauli observable, a measurement operator, or further information about the description of $O$ is provided in the high-privacy regime $(\varepsilon \leq 1)$ 
   \end{enumerate}}


This work showcases the importance of selecting the right utility metric in understanding privacy-utility tradeoffs. We focused on quantum local differential privacy (QLDP), which is suitable for applications where data is with several local parties, and offers a framework ensuring privacy under all possible adversaries. It would be interesting to study privacy-utility tradeoffs under flexible privacy frameworks, such as those used in~\cite{nuradha_QPP,measuredHS25}. 
An important technical tool towards this study would be sample complexity lower bounds on private quantum hypothesis testing when privacy is quantified by those flexible frameworks, in contrast to QLDP. Even in the QLDP setting, obtaining lower bounds that hold for $\delta >0$ would be useful in expanding the results obtained in this work on expectation estimation to a wide range of parameter regimes. To this effort, information contraction bounds under $(\varepsilon,\delta)$-QLDP developed recently in~\cite{nuradha2025nonLinearSDPI,dasgupta2025quantum} may be useful. 

In summary, this work serves both as a foundation and an invitation to study various quantum information processing tasks and estimate properties of quantum systems, while ensuring privacy for sensitive information such as the identity of the state.

\section*{Acknowledgments}
We thank Vishal Singh and Mark M.~Wilde for helpful discussions on privacy-utility tradeoffs in the quantum setting. TN acknowledges support from the IQUIST Postdoctoral Fellowship from the Illinois Quantum Information Science and Technology Center at the University of Illinois Urbana-Champaign.
SB and FL acknowledge support from National Science Foundation Grants No.~2426103 and 2442410.

\medskip
\changing{ \noindent AI Statement: During the peer review process while revising the paper, we used the assistance of ChatGPT 5.6 Sol to identify a mistake in the proof of~\Cref{prop:upper_2}. AI assistance was only used to identify the mistake. }

\printbibliography

@article{koenigsmolin,
  title     = {How to efficiently select an arbitrary Clifford group element},
  author    = {Koenig, Robert and Smolin, John A},
  journal   = {Journal of Mathematical Physics},
  volume    = {55},
  number    = {12},
  year      = {2014},
  publisher = {AIP Publishing}
}

@InProceedings{DMNS06,
  author        = {Dwork, Cynthia and McSherry, Frank and Nissim, Kobbi and Smith, Adam},
  title         = {Calibrating noise to sensitivity in private data analysis},
  booktitle     = {Proceedings of Conference on Theory of Cryptography, TCC},
  year          = {2006},
  pages         = {265-284},
  archiveprefix = {arXiv},
  doi           = {10.1007/11681878_14},
}

@Article{DR14,
  author     = {Dwork, Cynthia and Roth, Aaron},
  title      = {The Algorithmic Foundations of Differential Privacy},
  journal    = {Found. Trends Theor. Comput. Sci.},
  year       = {2014},
  volume     = {9},
  number     = {3–4},
  pages      = {211--407},
  month      = aug,
  issn       = {1551-305X},
  address    = {Hanover, MA, USA},
  doi        = {10.1561/0400000042},
  issue_date = {Aug 2014},
  numpages   = {197},
  publisher  = {Now Publishers Inc.},
  url        = {https://doi.org/10.1561/0400000042},
}

@InProceedings{CY16,
  author        = {Cuff, Paul and Yu, Lanqing},
  title         = {Differential Privacy as a Mutual Information Constraint},
  booktitle     = {Proceedings of the 2016 ACM SIGSAC Conference on Computer and Communications Security},
  year          = {2016},
  series        = {CCS '16},
  pages         = {43--54},
  publisher     = {Association for Computing Machinery},
  archiveprefix = {arXiv},
  doi           = {10.1145/2976749.2978308},
  eprint        = {1608.03677},
  isbn          = {9781450341394},
  location      = {Vienna, Austria},
  numpages      = {12},
  primaryclass  = {cs.IT},
  url           = {https://doi.org/10.1145/2976749.2978308},
}

@InProceedings{erlingsson2014rappor,
  author        = {Erlingsson, \'{U}lfar and Pihur, Vasyl and Korolova, Aleksandra},
  title         = {RAPPOR: Randomized Aggregatable Privacy-Preserving Ordinal Response},
  booktitle     = {Proceedings of the 2014 ACM SIGSAC Conference on Computer and Communications Security},
  year          = {2014},
  series        = {CCS '14},
  pages         = {1054--1067},
  publisher     = {Association for Computing Machinery},
  archiveprefix = {arXiv},
  doi           = {10.1145/2660267.2660348},
  eprint        = {1407.6981},
  isbn          = {9781450329576},
  location      = {Scottsdale, Arizona, USA},
  numpages      = {14},
  primaryclass  = {cs.CR},
  url           = {https://doi.org/10.1145/2660267.2660348},
}

@InProceedings{kamath2019privately,
  author        = {Kamath, Gautam and Li, Jerry and Singhal, Vikrant and Ullman, Jonathan},
  title         = {Privately Learning High-Dimensional Distributions},
  booktitle     = {Proceedings of the Thirty-Second Conference on Learning Theory},
  year          = {2019},
  editor        = {Beygelzimer, Alina and Hsu, Daniel},
  volume        = {99},
  series        = {Proceedings of Machine Learning Research},
  pages         = {1853--1902},
  month         = {6},
  publisher     = {PMLR},
  archiveprefix = {arXiv},
  eprint        = {1805.00216},
  primaryclass  = {cs.DS},
  url           = {https://proceedings.mlr.press/v99/kamath19a.html},
}

@InProceedings{kamath2020private,
  author        = {Kamath, Gautam and Singhal, Vikrant and Ullman, Jonathan},
  title         = {Private Mean Estimation of Heavy-Tailed Distributions},
  booktitle     = {Proceedings of Thirty Third Conference on Learning Theory},
  year          = {2020},
  editor        = {Abernethy, Jacob and Agarwal, Shivani},
  volume        = {125},
  series        = {Proceedings of Machine Learning Research},
  pages         = {2204--2235},
  month         = {7},
  publisher     = {PMLR},
  archiveprefix = {arXiv},
  eprint        = {2002.09464},
  primaryclass  = {cs.DS},
  url           = {https://proceedings.mlr.press/v125/kamath20a.html},
}

@InProceedings{structure_optimal_16,
  author        = {Canonne, Cl\'{e}ment L. and Kamath, Gautam and McMillan, Audra and Smith, Adam and Ullman, Jonathan},
  title         = {The structure of optimal private tests for simple hypotheses},
  booktitle     = {Proceedings of the Annual ACM SIGACT Symposium on Theory of Computing},
  year          = {2019},
  pages         = {310--321},
  archiveprefix = {arXiv},
  doi           = {10.1145/3313276.3316336},
  eprint        = {1811.11148},
  isbn          = {9781450367059},
  numpages      = {12},
  primaryclass  = {cs.DS},
}

@InProceedings{murakami2019utilityOptimizedLDP,
  author        = {Takao Murakami and Yusuke Kawamoto},
  title         = {{Utility-Optimized} Local Differential Privacy Mechanisms for Distribution Estimation},
  booktitle     = {28th USENIX Security Symposium (USENIX Security 19)},
  year          = {2019},
  pages         = {1877--1894},
  address       = {Santa Clara, CA},
  month         = {8},
  publisher     = {USENIX Association},
  archiveprefix = {arXiv},
  eprint        = {1807.11317},
  isbn          = {978-1-939133-06-9},
  primaryclass  = {cs.DB},
  url           = {https://www.usenix.org/conference/usenixsecurity19/presentation/murakami},
}

@Article{rassouli2019optimalPrivacyUtilityClassical,
  author        = {Rassouli, Borzoo and Gündüz, Deniz},
  title         = {Optimal Utility-Privacy Trade-Off With Total Variation Distance as a Privacy Measure},
  journal       = {IEEE Transactions on Information Forensics and Security},
  year          = {2020},
  volume        = {15},
  pages         = {594-603},
  archiveprefix = {arXiv},
  doi           = {10.1109/TIFS.2019.2903658},
  eprint        = {1801.02505},
  primaryclass  = {cs.IT},
}

@InProceedings{Ankit2023simple,
  author        = {Pensia, Ankit and Asadi, Amir Reza and Jog, Varun and Loh, Po-Ling},
  title         = {Simple Binary Hypothesis Testing under Local Differential Privacy and Communication Constraints},
  booktitle     = {Proceedings of Thirty Sixth Conference on Learning Theory},
  year          = {2023},
  editor        = {Neu, Gergely and Rosasco, Lorenzo},
  volume        = {195},
  series        = {Proceedings of Machine Learning Research},
  pages         = {3229--3230},
  month         = {7},
  publisher     = {PMLR},
  archiveprefix = {arXiv},
  eprint        = {2301.03566},
  primaryclass  = {math.ST},
  url           = {https://proceedings.mlr.press/v195/pensia23a.html},
}

@Article{wang2017privacy,
  author        = {Wang, Ye and Basciftci, Yuksel Ozan and Ishwar, Prakash},
  title         = {Privacy-utility tradeoffs under constrained data release mechanisms},
  journal       = {arXiv preprint},
  year          = {2017},
  archiveprefix = {arXiv},
  eprint        = {1710.09295},
}

@Article{Contraction_local_new24,
  author        = {Asoodeh, Shahab and Zhang, Huanyu},
  title         = {Contraction of Locally Differentially Private Mechanisms},
  journal       = {IEEE Journal on Selected Areas in Information Theory},
  year          = {2024},
  volume        = {5},
  pages         = {385-395},
  archiveprefix = {arXiv},
  doi           = {10.1109/JSAIT.2024.3397305},
  eprint        = {2210.13386},
  primaryclass  = {cs.IT},
}

@InProceedings{mironov2017renyi,
  author        = {Mironov, Ilya},
  title         = {Rényi Differential Privacy},
  booktitle     = {2017 IEEE 30th Computer Security Foundations Symposium (CSF)},
  year          = {2017},
  pages         = {263-275},
  archiveprefix = {arXiv},
  doi           = {10.1109/CSF.2017.11},
  eprint        = {1702.07476},
  primaryclass  = {cs.CR},
}

@InProceedings{kairouz2014extremal,
  author        = {Kairouz, Peter and Oh, Sewoong and Viswanath, Pramod},
  title         = {Extremal Mechanisms for Local Differential Privacy},
  booktitle     = {Advances in Neural Information Processing Systems},
  year          = {2014},
  volume        = {27},
  pages         = {1-9},
  archiveprefix = {arXiv},
  eprint        = {1407.1338},
  primaryclass  = {cs.IT},
}

@Article{duchi2018minimax,
  author        = {John C. Duchi and Michael I. Jordan and Martin J. Wainwright},
  title         = {Minimax Optimal Procedures for Locally Private Estimation},
  journal       = {Journal of the American Statistical Association},
  year          = {2018},
  volume        = {113},
  number        = {521},
  pages         = {182--201},
  archiveprefix = {arXiv},
  doi           = {10.1080/01621459.2017.1389735},
  eprint        = {1604.02390},
  primaryclass  = {math.ST},
  publisher     = {Taylor \& Francis},
}

@InProceedings{DJW13,
  author        = {Duchi, John C. and Jordan, Michael I. and Wainwright, Martin J.},
  title         = {Local Privacy and Statistical Minimax Rates},
  booktitle     = {2013 IEEE 54th Annual Symposium on Foundations of Computer Science},
  year          = {2013},
  pages         = {429-438},
  archiveprefix = {arXiv},
  doi           = {10.1109/FOCS.2013.53},
  eprint        = {1302.3203},
  primaryclass  = {math.ST},
}

@Article{asoodeh2021local,
  author        = {Asoodeh, Shahab and Aliakbarpour, Maryam and Calmon, Flavio P},
  title         = {Local Differential Privacy Is Equivalent to Contraction of ${E}_\gamma$-Divergence},
  journal       = {arXiv preprint},
  year          = {2021},
  archiveprefix = {arXiv},
  eprint        = {2102.01258},
  primaryclass  = {cs.IT},
}

@InProceedings{zhao2020notAccuracyLoss,
  author        = {Zhao, Benjamin Zi Hao and Kaafar, Mohamed Ali and Kourtellis, Nicolas},
  title         = {Not one but many Tradeoffs: Privacy Vs. Utility in Differentially Private Machine Learning},
  booktitle     = {Proceedings of the 2020 ACM SIGSAC Conference on Cloud Computing Security Workshop},
  year          = {2020},
  series        = {CCSW'20},
  pages         = {15--26},
  publisher     = {Association for Computing Machinery},
  archiveprefix = {arXiv},
  doi           = {10.1145/3411495.3421352},
  eprint        = {2008.08807},
  isbn          = {9781450380843},
  location      = {Virtual Event, USA},
  numpages      = {12},
  primaryclass  = {cs.CR},
  url           = {https://doi.org/10.1145/3411495.3421352},
}

@InProceedings{SKT18,
  author        = {Showkatbakhsh, Mehrdad and Karakus, Can and Diggavi, Suhas},
  title         = {Privacy-Utility Trade-off of Linear Regression under Random Projections and Additive Noise},
  booktitle     = {IEEE International Symposium on Information Theory (ISIT)},
  year          = {2018},
  pages         = {186-190},
  archiveprefix = {arXiv},
  doi           = {10.1109/ISIT.2018.8437722},
  eprint        = {1902.04688},
  primaryclass  = {cs.LG},
}

@Book{SoK_DP22,
  title         = {Guide to Differential Privacy Modifications: A Taxonomy of Variants and Extensions},
  publisher     = {Springer Briefs in Computer Science Series. Springer International Publishing AG, Cham},
  year          = {2022},
  author        = {Pej{\'o}, Bal{\'a}zs and Desfontaines, Damien},
  archiveprefix = {arXiv},
}

@InProceedings{aaronson2019gentle,
  author        = {Aaronson, Scott and Rothblum, Guy N.},
  title         = {Gentle measurement of quantum states and differential privacy},
  booktitle     = {Proceedings of the 51st Annual ACM SIGACT Symposium on Theory of Computing},
  year          = {2019},
  series        = {STOC 2019},
  pages         = {322--333},
  publisher     = {Association for Computing Machinery},
  archiveprefix = {arXiv},
  doi           = {10.1145/3313276.3316378},
  eprint        = {1904.08747},
  isbn          = {9781450367059},
  location      = {Phoenix, AZ, USA},
  numpages      = {12},
  primaryclass  = {quant-ph},
  url           = {https://doi.org/10.1145/3313276.3316378},
}

@Article{angrisani_localModel25,
  author        = {Angrisani, Armando and Kashefi, Elham},
  title         = {Quantum Differential Privacy in the Local Model},
  journal       = {IEEE Transactions on Information Theory},
  year          = {2025},
  volume        = {71},
  number        = {5},
  pages         = {3675-3692},
  archiveprefix = {arXiv},
  doi           = {10.1109/TIT.2025.3552671},
}

@Article{nuradha2025nonLinearSDPI,
  author        = {Nuradha, Theshani and George, Ian and Hirche, Christoph},
  title         = {Non-Linear Strong Data-Processing for Quantum Hockey-Stick Divergences},
  journal       = {arXiv preprint},
  year          = {2025},
  archiveprefix = {arXiv},
  eprint        = {2512.16778},
}

@InProceedings{Farhad_QP_HT,
  author        = {Farokhi, Farhad},
  title         = {Quantum Privacy and Hypothesis-Testing},
  booktitle     = {2023 62nd IEEE Conference on Decision and Control (CDC)},
  year          = {2023},
  pages         = {2841-2846},
  archiveprefix = {arXiv},
  doi           = {10.1109/CDC49753.2023.10384192},
  eprint        = {2302.12405},
  primaryclass  = {quant-ph},
}

@Article{alabi2026QDPComposition,
  author        = {Alabi, Daniel and Nuradha, Theshani},
  title         = {When Does Quantum Differential Privacy Compose?},
  journal       = {arXiv preprint},
  year          = {2026},
  archiveprefix = {arXiv},
  eprint        = {2601.00337},
  primaryclass  = {quant-ph},
}

@Article{dasgupta2025quantum,
  author        = {Dasgupta, Ayanava and Warsi, Naqueeb Ahmad and Hayashi, Masahito},
  title         = {Quantum Information Ordering and Differential Privacy},
  journal       = {arXiv preprint},
  year          = {2025},
  archiveprefix = {arXiv},
  eprint        = {2511.01467},
  primaryclass  = {quant-ph},
}

@InProceedings{measuredHS25,
  author        = {Nuradha, Theshani and Singh, Vishal and Wilde, Mark M.},
  title         = {Measured Hockey-Stick Divergence and Its Applications to Quantum Pufferfish Privacy},
  booktitle     = {2025 IEEE International Symposium on Information Theory (ISIT)},
  year          = {2025},
  pages         = {1-6},
  archiveprefix = {arXiv},
  doi           = {10.1109/ISIT63088.2025.11195501},
  eprint        = {2501.12359},
  primaryclass  = {quant-ph},
}

@Article{SW12,
  author        = {Sharma, Naresh and Warsi, Naqueeb Ahmad},
  title         = {On the strong converses for the quantum channel capacity theorems},
  journal       = {arXiv preprint},
  year          = {2012},
  archiveprefix = {arXiv},
  eprint        = {1205.1712},
}

@Article{huang2022learning,
  author        = {Huang, Hsin-Yuan},
  title         = {Learning quantum states from their classical shadows},
  journal       = {Nature Reviews Physics},
  year          = {2022},
  volume        = {4},
  number        = {2},
  pages         = {81--81},
  archiveprefix = {arXiv},
  doi           = {10.1038/s42254-021-00411-5},
  publisher     = {Nature Publishing Group UK London},
}

@Article{NuradhaQueryComplexity_2025,
  author        = {Nuradha, Theshani and Wilde, Mark M},
  title         = {Query complexity of classical and quantum channel discrimination},
  journal       = {Quantum Science and Technology},
  year          = {2025},
  volume        = {10},
  number        = {4},
  pages         = {045075},
  month         = {10},
  note          = {arXiv:2504.12989 [quant-ph]},
  archiveprefix = {arXiv},
  doi           = {10.1088/2058-9565/ae0a79},
  publisher     = {IOP Publishing},
}

@Article{cheng2024invitation,
  author        = {Cheng, Hao-Chung and Datta, Nilanjana and Liu, Nana and Nuradha, Theshani and Salzmann, Robert and Wilde, Mark M.},
  title         = {An invitation to the sample complexity of quantum hypothesis testing},
  journal       = {npj Quantum Information},
  year          = {2025},
  volume        = {11},
  number        = {1},
  pages         = {94},
  issn          = {2056-6387},
  archiveprefix = {arXiv},
  doi           = {10.1038/s41534-025-00980-8},
  eprint        = {2403.17868},
  primaryclass  = {quant-ph},
}

@Article{Koh2022classicalshadows,
  author        = {Koh, Dax Enshan and Grewal, Sabee},
  title         = {Classical shadows with noise},
  journal       = {{Quantum}},
  year          = {2022},
  volume        = {6},
  pages         = {776},
  month         = Aug,
  issn          = {2521-327X},
  archiveprefix = {arXiv},
  doi           = {10.22331/q-2022-08-16-776},
  eprint        = {2011.11580},
  primaryclass  = {quant-ph},
  publisher     = {{Verein zur F{\"{o}}rderung des Open Access Publizierens in den Quantenwissenschaften}},
}

@Article{huang2020predicting,
  author        = {Huang, Hsin-Yuan and Kueng, Richard and Preskill, John},
  title         = {Predicting many properties of a quantum system from very few measurements},
  journal       = {Nature Physics},
  year          = {2020},
  volume        = {16},
  number        = {10},
  pages         = {1050--1057},
  archiveprefix = {arXiv},
  doi           = {10.1038/s41567-020-0932-7},
  eprint        = {2002.08953},
  primaryclass  = {quant-ph},
  publisher     = {Nature Publishing Group UK London},
}

@PhdThesis{gallage2025theory,
  author        = {Gallage, Theshani Nuradha Piliththuwasam},
  title         = {Theory of Privacy and Testing in a Quantum World},
  school        = {Cornell University},
  year          = {2025},
  archiveprefix = {arXiv},
  url           = {https://www.proquest.com/openview/2cb580e718241481b84185ba6d289ce7/},
}

@Article{Uhl76_nourl,
  author  = {A. Uhlmann},
  title   = {The “transition probability” in the state space of a $*$-algebra},
  journal = {Reports on Mathematical Physics},
  year    = {1976},
  volume  = {9},
  number  = {2},
  pages   = {273-279},
  issn    = {0034-4877},
  doi     = {https://doi.org/10.1016/0034-4877(76)90060-4},
  url     = {https://www.sciencedirect.com/science/article/pii/0034487776900604},
}

@Article{angrisani2023unifying,
  author        = {Angrisani, Armando and Doosti, Mina and Kashefi, Elham},
  title         = {A unifying framework for differentially private quantum algorithms},
  journal       = {arXiv preprint},
  year          = {2023},
  archiveprefix = {arXiv},
  eprint        = {2307.0473},
  primaryclass  = {quant-ph},
}

@Article{hirche2023quantum,
  author        = {Hirche, Christoph and Rouz{\'e}, Cambyse and Fran{\c{c}}a, Daniel Stilck},
  title         = {Quantum Differential Privacy: An Information Theory Perspective},
  journal       = {IEEE Transactions on Information Theory},
  year          = {2023},
  volume        = {69},
  number        = {9},
  pages         = {5771-5787},
  archiveprefix = {arXiv},
  doi           = {10.1109/TIT.2023.3272904},
  eprint        = {2202.10717},
  primaryclass  = {quant-ph},
}

@Article{nuradha_QPP,
  author        = {Nuradha, Theshani and Goldfeld, Ziv and Wilde, Mark M.},
  title         = {Quantum Pufferfish Privacy: A Flexible Privacy Framework for Quantum Systems},
  journal       = {IEEE Transactions on Information Theory},
  year          = {2024},
  volume        = {70},
  number        = {8},
  pages         = {5731-5762},
  archiveprefix = {arXiv},
  doi           = {10.1109/TIT.2024.3404927},
  eprint        = {2306.13054},
  primaryclass  = {quant-ph},
}

@InProceedings{QDP_computation17,
  author    = {Zhou, Li and Ying, Mingsheng},
  title     = {Differential Privacy in Quantum Computation},
  booktitle = {2017 IEEE 30th Computer Security Foundations Symposium (CSF)},
  year      = {2017},
  pages     = {249-262},
  doi       = {10.1109/CSF.2017.23},
}

@Article{nuradha2022pufferfishJ,
  author        = {Nuradha, Theshani and Goldfeld, Ziv},
  title         = {Pufferfish Privacy: An Information-Theoretic Study},
  journal       = {IEEE Transactions on Information Theory},
  year          = {2023},
  volume        = {69},
  number        = {11},
  pages         = {7336-7356},
  archiveprefix = {arXiv},
  doi           = {10.1109/TIT.2023.3296288},
  eprint        = {2210.12612},
  primaryclass  = {cs.IT},
}

@Article{Christoph2024sample,
  author        = {Cheng, Hao-Chung and Hirche, Christoph and Rouz{\'e}, Cambyse},
  title         = {Sample Complexity of Locally Differentially Private Quantum Hypothesis Testing},
  journal       = {arXiv preprint},
  year          = {2024},
  archiveprefix = {arXiv},
  eprint        = {2406.18658},
  primaryclass  = {quant-ph},
}

@Article{KM14,
  author     = {Kifer, Daniel and Machanavajjhala, Ashwin},
  title      = {Pufferfish: A framework for mathematical privacy definitions},
  journal    = {ACM Trans. Database Syst.},
  year       = {2014},
  volume     = {39},
  number     = {1},
  month      = jan,
  issn       = {0362-5915},
  address    = {New York, NY, USA},
  articleno  = {3},
  doi        = {10.1145/2514689},
  issue_date = {January 2014},
  numpages   = {36},
  publisher  = {Association for Computing Machinery},
  url        = {https://doi.org/10.1145/2514689},
}

@Article{nuradha2024contraction,
  author        = {Nuradha, Theshani and Wilde, Mark M.},
  title         = {Contraction of Private Quantum Channels and Private Quantum Hypothesis Testing},
  journal       = {IEEE Transactions on Information Theory},
  year          = {2025},
  volume        = {71},
  number        = {3},
  pages         = {1851-1873},
  archiveprefix = {arXiv},
  doi           = {10.1109/TIT.2025.3527859},
  eprint        = {2406.18651},
  primaryclass  = {quant-ph},
}

@InProceedings{guan2024optimal,
  author        = {Guan, Ji},
  title         = {Optimal Mechanisms for Quantum Local Differential Privacy},
  booktitle     = {Proceedings of the 2025 ACM SIGSAC Conference on Computer and Communications Security},
  year          = {2025},
  pages         = {3737--3749},
  archiveprefix = {arXiv},
  doi           = {10.1145/3719027.3765178},
  eprint        = {2407.13516},
  isbn          = {9798400715259},
  numpages      = {13},
  primaryclass  = {quant-ph},
}

@InProceedings{DP_testing18,
  author        = {Acharya, Jayadev and Sun, Ziteng and Zhang, Huanyu},
  title         = {Differentially Private Testing of Identity and Closeness of Discrete Distributions},
  booktitle     = {Advances in Neural Information Processing Systems},
  year          = {2018},
  editor        = {S. Bengio and H. Wallach and H. Larochelle and K. Grauman and N. Cesa-Bianchi and R. Garnett},
  volume        = {31},
  publisher     = {Curran Associates, Inc.},
  archiveprefix = {arXiv},
  eprint        = {1707.05128},
  primaryclass  = {cs.LG},
  url           = {https://proceedings.neurips.cc/paper_files/paper/2018/file/7de32147a4f1055bed9e4faf3485a84d-Paper.pdf},
}

@Article{namPUT_HT_2025quantum,
  author        = {Nam, Seung-Hyun and Park, Hyun-Young and Bae, Joonwoo and Lee, Si-Hyeon},
  title         = {Quantum Advantage in Private Multiple Hypothesis Testing},
  journal       = {arXiv preprint},
  year          = {2025},
  archiveprefix = {arXiv},
  eprint        = {2501.10152},
  primaryclass  = {quant-ph},
}

\appendix


\subsection{Proof of \Cref{prop:optimal_P_QLDP_Dep}} \label{App:optimal_p_QLDP}

    First, recall that $\cA$ being $(\varepsilon, \delta)$-QLDP is equivalent to~\eqref{eq:equivalence_pure}:
    \begin{equation} \label{eq:QLDP_Equiv}
    \sup_{\rho,\sigma \in \cD(\cH)}  E_{e^\varepsilon}\!\left( \cA(\rho) \Vert \cA(\sigma) \right)  = \sup_{\varphi_{1}\perp\varphi_{2}}E_{e^\varepsilon}(\mathcal{A}(\varphi_{1}
)\Vert\mathcal{A}(\varphi_{2})) \leq \delta,
\end{equation}
where $\varphi_{1}$ and $\varphi_{2}$ are orthogonal pure
states.

With that consider for $\gamma \geq 1$, 
\begin{align}
    & E_{\gamma}(\mathcal{A}_p(\varphi_{1}
)\Vert\mathcal{A}_p(\varphi_{2})) \notag \\
& =\Tr\!\left[ \left( \cA_p(\varphi_1) - \gamma \cA_p(\varphi_2) \right)_+\right] \\ 
&= \Tr\!\left[ \left((1-p) (\varphi_1 -\gamma \varphi_2) + \frac{pI}{d} (1-\gamma) \right)_+ \right] \\ 
&= \Tr\!\Bigg[ \Bigg( \varphi_1 \left(  1-p + \frac{p(1-\gamma)}{d} \right) + \varphi_2 \left(-\gamma( 1-p) + \frac{p(1-\gamma)}{d}  \right) \nonumber \\ 
& \quad \quad \quad  \quad \hspace{30mm} +(I- \varphi_1 -\varphi_2) \left(   \frac{p(1-\gamma)}{d} \right) \Bigg)_+ \Bigg] \\
&=\left(  1-p + \frac{p(1-\gamma)}{d} \right)_+,
\end{align}
where the last equality follows since $\varphi_1, \varphi_2, I- \varphi_1-\varphi_2$ are orthogonal projections and due to $\gamma \geq 1$, the positive part is contributed only by the projection $\varphi_1$.

Using the above equality, we have 
\begin{equation}
    \sup_{\varphi_{1}\perp\varphi_{2}}E_{\gamma}(\mathcal{A}_p(\varphi_{1}
)\Vert\mathcal{A}_p(\varphi_{2})) = \left(  1-p + \frac{p(1-\gamma)}{d} \right)_+ =\left(  1 -p \left(\frac{d-1+\gamma}{d}\right) \right)_+
\end{equation}
Thus, when $\gamma \geq 1$, we see that 
\begin{equation}
    p \geq \frac{d(1-\delta)}{\gamma +d-1} \implies  \sup_{\varphi_{1}\perp\varphi_{2}}E_{\gamma}(\mathcal{A}_p(\varphi_{1}
)\Vert\mathcal{A}_p(\varphi_{2})) \leq \delta,
\end{equation}
and 
\begin{equation}
    p < \frac{d(1-\delta)}{\gamma +d-1} \implies  \sup_{\varphi_{1}\perp\varphi_{2}}E_{\gamma}(\mathcal{A}_p(\varphi_{1}
)\Vert\mathcal{A}_p(\varphi_{2})) > \delta.
\end{equation}
Finally, we conclude the proposition statement by recalling the equivalence to $(\varepsilon,\delta)$-QLDP as in~\eqref{eq:QLDP_Equiv} by choosing $\gamma=e^\varepsilon \geq 1$.

\subsection{Proof of \Cref{prop:N_u_N_G_private}} \label{App:prop_N_U_N_G_private}

 Consider that 
    \begin{align}
        E_\gamma\!\left( \cN_U(\rho) \Vert \cN_U(\sigma)\right) & = E_\gamma\!\left(  U^\dag \cN( U \rho U^\dag) U \Vert U^\dag \cN( U \sigma U^\dag) U\right) \\
        &=  E_\gamma\!\left(  \cN( U \rho U^\dag) \Vert \cN( U \sigma U^\dag) \right), \label{eq:unitary-rotated-input}
    \end{align}
    where the last line follows by the unitary invariance of $E_\gamma$, itself a consequence of data processing.
    It follows from \eqref{eq:unitary-rotated-input} that $E_\gamma(\cN(\cdot)\|\cN(\cdot))$ and $E_\gamma(\cN(U\cdot U^\dagger)\|\cN(U\cdot U^\dagger))$ have the same supremum over all pairs of states $(\rho,\sigma)$, and hence $\cN$ is $(\varepsilon,\delta)$-QLDP if and only if $\cN_U$ is.

For the second assertion, recall that $\cN_G = \int_{G} \cV_{g^{-1}}\circ\cN\circ\cU_g \,\mathrm d\mu(g)$, where $\cV_g=V_g\cdot V_g^\dagger$ and $\cU_g=U_g\cdot U_g^\dagger$ for $g\in G$ are the unitary channels associated with the chosen unitary representations.
Then, using the joint convexity of $E_\gamma$ (see \cite{hirche2023quantum}), we have for all $\rho,\sigma$ that
\begin{align}
    E_\gamma(\cN_G(\rho)\|\cN_G(\sigma)) &\leq \int_{G} E_\gamma(V_g^\dagger \cN(U_g\rho U_g^\dagger) V_g \| V_g^\dagger \cN(U_g\sigma U_g^\dagger) V_g) \,\mathrm d\mu(g)\\
    &= \int_{G} E_\gamma(\cN_{U_g}(\rho) \| \cN_{U_g}(\sigma)) \,\mathrm d\mu(g). \label{eq:twirled-channel-hockey-stick}
\end{align}
Taking suprema over pairs $(\rho,\sigma)$ on both sides of \eqref{eq:twirled-channel-hockey-stick} and using an argument similar to that in the first part of the proposition, we see that $\cN_G$ is $(\varepsilon, \delta)$-QLDP if $\cN$ is.
\subsection{Proof of \Cref{lem:twirl-fidelity-qdps}} \label{app:Lemma_twirl_trace}

For fidelity:     Note first that, for all unitaries $U$, 
    \begin{align}
        F(\cN_U) = \min_{\rho} F(U^\dagger\cN(U\rho U^\dagger) U, \rho) = \min_\rho F(\cN(U\rho U^\dagger), U\rho U^\dagger) = F(\cN),\label{eq:fidelity-invariance}
    \end{align}
    where the second equality follows from the unitary invariance of fidelity.
    Using concavity of the (squared) fidelity in one argument, we then have for any $\rho$ that
    \begin{align}
    F\left( \int_G\mathrm{d}\mu(g)\, U_g^\dagger \cN(U_g \rho U_g^\dagger) U_g, \rho\right) &\geq \int_G\mathrm{d}\mu(g)\,F\left(U_g^\dagger \cN(U_g \rho U_g^\dagger) U_g, \rho \right) \\  &\geq  \int_G\mathrm{d}\mu(g)\, F(\cN_{U_g}) = F(\cN),
    \end{align}
    where we used \eqref{eq:fidelity-invariance} in the equality.
    In particular, this holds for the state $\rho$ achieving the minimum in $F(\cN_G)$, which concludes the proof.

For trace distance: First, we note that $T(\cN) = T(\cN_U)$ for any unitary $U$ by unitary invariance of the trace norm.
    Second, we have
    \begin{align}
        \left\| \int_G \mathrm{d}\mu(g)\, \cN_{U_g}(\rho) - \rho\right\|_1 &= \left\| \int_G \mathrm{d}\mu(g)\, U_g^\dagger \left(\cN(U_g\rho U_g^\dagger) - U_g\rho U_g^\dagger\right) U_g \right\|_1 \\
        &\leq \int_G \mathrm{d}\mu(g)\, \left\|U_g^\dagger \left(\cN(U_g\rho U_g^\dagger) - U_g\rho U_g^\dagger\right) U_g \right\|_1\\
        &= \int_G \mathrm{d}\mu(g)\, \left\|\cN_{U_g}(\rho)-\rho \right\|_1\\
        &\leq 2\int_G \mathrm{d}\mu(g)\, T(\cN_{U_g})\\
        &= 2 T(\cN),
    \end{align}
    where we used convexity of the trace norm in the first inequality.

\subsection{Proof of~\cref{thm:optimal-utility}} \label{app:proof_optimal_utility_generic_d}

  We first observe that, for a depolarizing channel $\cA_p$, we have
    \begin{align}
        T(\cA_{p}) &= \frac{p(d-1)}{d} & F(\cA_p) &= 1-\frac{p(d-1)}{d}.
        \label{eq:utility-dep-channel-formulas}
    \end{align}
    To see this, note that the trace norm is convex, and hence the optimum in $T(\cN)=\max_\rho T(\cN(\rho),\rho)$ is achieved on a pure state.
    Furthermore, for any pure state $\psi$,
    \begin{align}
     T\left( \cA_p( \psi ), \psi\right) &= \frac{1}{2}\left\| (1-p)\psi + p \frac{1}{d}I-\psi \right\|_1 \\ &= \frac{p}{2} \left\| \frac{1}{d} I - \psi\right\|_1 \\ &= \frac{p}{2}\left( 1-\frac{1}{d} + \frac{d-1}{d} \right) = \frac{p(d-1)}{d}.
    \end{align}
    The formula for $F(\cA_p)$ can be shown similarly together with the concavity of the fidelity and 
    \begin{align}
    F\left( \cA_p( \psi ), \psi\right) &= \langle \psi | \cA_p( \psi ) |\psi \rangle \\
    &= \left \langle \psi \middle | \left[(1-p) |\psi \rangle\!\langle \psi| + \frac{p I}{d} \right]   \middle |\psi \right \rangle \\ 
    &= (1-p) | \langle \psi | \psi \rangle |^2 + \frac{p}{d} \langle \psi | \psi \rangle  \\
    &=(1-p)+\frac{p}{d}.
\end{align} 
Also note that $F(\cA_p) =1-T(\cA_p)$.

    We now prove $\max_{\cN \in \cB^{\varepsilon, \delta}_d} F(\cN) =   \frac{e^\varepsilon +\delta (d-1)}{e^\varepsilon +d-1} $
    in the following; $\min_{\cN \in \cB^{\varepsilon, \delta}_d} T(\cN) =  \frac{(d-1)(1-\delta)}{e^\varepsilon +d-1}$ can be proved using analogous arguments.
    By \Cref{prop:optimal_P_QLDP_Dep} the set $\mathcal{B}^{\varepsilon,\delta}_d$ includes all depolarizing channels $\cA_p$ with
    \begin{align}
        p\geq \frac{d(1-\delta)}{e^\varepsilon + d-1}\eqqcolon p^*,\label{eq:p-condition}
    \end{align}
    and we have $F(\cA_p) = 1-p(1-1/d)$ by \eqref{eq:utility-dep-channel-formulas}, which is maximized by choosing $p=p^*$.
    Hence, 
    \begin{align} \max_{\cN\in\mathcal{B}^{\varepsilon,\delta}_d} F(\cN) \geq 1-p^*(1-1/d) = \frac{e^\varepsilon + \delta(d-1)}{e^\varepsilon + d-1}. \label{eq:F-LB}
    \end{align}
    
    Let now $\cM$ be a channel maximizing $\max_{\cN\in\mathcal{B}^{\varepsilon,\delta}_d} F(\cN)$. The twirled channel $\cM_{\mathbb{U}(d)}$ is a depolarizing channel $\cA_q$ with $q = q(\cM)$ also satisfying \eqref{eq:p-condition} by \Cref{prop:N_u_N_G_private}.
    Using \Cref{lem:twirl-fidelity-qdps}, we then have
    \begin{align}
        \max_{\cN\in\mathcal{B}^{\varepsilon,\delta}_d} F(\cN) \equiv F(\cM) \leq F\left(\cM_{\mathbb{U}(d)}\right) = 1-q(1-1/d) \leq 1-p^*(1-1/d), \label{eq:F-UB}
    \end{align}
    where the last inequality follows from the fact that $p\mapsto 1-p(1-1/d)$ is non-increasing and $q\geq p^*$.
    \Cref{eq:F-LB,eq:F-UB} together show that $  \max_{\cN\in\mathcal{B}^{\varepsilon,\delta}_d} F(\cN) = \frac{e^\varepsilon + \delta(d-1)}{e^\varepsilon + d-1},$
    which concludes the proof.

    \subsection{Private Quantum Hypothesis Testing} \label{App:Private_QHT}

    For $\delta=0$, let $\cA$ be an $\varepsilon$-QLDP mechanism. Then
\begin{align} \label{eq:def_SC}
    \mathrm{SC}^{\cA} _{(\rho,\sigma)}(\alpha,p,q) \coloneqq  
    & \min \!\left\{ n \in \mathbb{N}: p_{e}\!\left((\cA(\rho))^{\otimes n},(\cA(\sigma))^{\otimes n},p,q \right) \leq \alpha \right\},
\end{align}
where \begin{equation}
p_{e}\!\left(\rho^{\otimes n},\sigma^{\otimes n},p,q \right) \coloneqq \frac{1}{2}\left(  1-\left\Vert p\rho^{\otimes n}-q\sigma^{\otimes
n}\right\Vert _{1}\right).\label{eq:eps-n-relation}%
\end{equation}
The sample complexity under $\varepsilon$-QLDP is defined as follows: 
\begin{equation}\label{eq:SC_private_int}
    \mathrm{SC}^{\varepsilon} _{(\rho,\sigma)} (\alpha,p,q)\coloneqq \inf_{\cA \in \cB^\varepsilon} \mathrm{SC}^{\cA} _{(\rho,\sigma)}(\alpha,p,q), 
\end{equation}
where $\cB^\varepsilon$ is the set of all $\varepsilon$-QLDP mechanisms defined in \eqref{eq:private-channels} with $\delta=0$. 

With that, we have the following result, which is a combination of private QHT bounds from~\cite{nuradha2024contraction,Christoph2024sample} (see Theorem~8 of ~\cite{gallage2025theory}), together with the use of an improved lower bound on the sample complexity of QHT obtained in~\cite[Theorem~2]{NuradhaQueryComplexity_2025}, instead of the first lower bound in~\cite[Theorem~II.7]{cheng2024invitation}. 
\begin{theorem}[Bounds on Private Sample Complexity
]\label{thm:bounds_sample_C_private}
    Let $p\in(0,1)$, set $q\coloneqq 1-p$, and let $\rho$ and $\sigma$ be states. Fix the error probability $\alpha \in (0,pq)$. For $\varepsilon>0$, the following holds:
    \begin{align}
      \max \left\{\frac{C_{\varepsilon,p,q,\alpha}}{T(\rho,\sigma)} , \frac{\ln\!\left(\frac{pq}{\alpha (1-\alpha)} \right) e^\varepsilon }{ 2(e^\varepsilon -1)^2  \left[T(\rho,\sigma)\right]^2}  \right \} 
 \leq \mathrm{SC}^{\varepsilon} _{(\rho,\sigma)} (\alpha,p,q) 
  \leq  \left \lceil 2 \ln\!\left( \frac{\sqrt{pq}}{\alpha} \right) \!\left( \frac{(e^\varepsilon +1)}{(e^\varepsilon -1) T(\rho, \sigma)} \right)^2  \right\rceil,
    \label{eq:private-sample-compl-bnds}
    \end{align}
    where $C_{\varepsilon,p,q,\alpha} \coloneqq   \max\!\left\{ \frac{\ln\!\left(\frac{pq}{\alpha (1-\alpha)} \right) (e^\varepsilon +1)}{\varepsilon (e^\varepsilon -1)  } ,\frac{\left({1-\frac{\alpha(1-\alpha)}{pq}} \right) (e^\varepsilon +1)}{2 \ (e^{\varepsilon/2}-1)^2  } \right\}.$
\end{theorem}

In fact, this also concludes that for fixed $\alpha,p,q$ with $\alpha \leq pq$ and for $\varepsilon \in (0,1]$, we have that \cite[Theorem~9]{gallage2025theory}, 
\begin{equation}
    \mathrm{SC}^{\varepsilon} _{(\rho,\sigma)} (\alpha,p,q) = \Theta\!\left(\frac{1}{{\varepsilon^2 \left[T(\rho,\sigma) \right]^2}}\right).
\end{equation}

Note that the lower bounds in~\eqref{eq:private-sample-compl-bnds} are providing limits on any hypothesis testing protocol that gets access to privatized states instead of noiseless states. With the use of the above result, we obtain a lower bound on the sample complexity of the observable estimation task as given in the proof of \Cref{prop:LB_1}.

\subsection{Proof of \Cref{prop:LB_1}} \label{app:estimation_LB_proof}

 Let $|\psi_{\max} \rangle $ and $|\psi_{\min} \rangle $ be (normalized) eigenvectors corresponding to the maximum and minimum eigenvalues $\lambda_{\max}(O)$ and $\lambda_{\min}(O)$ of the observable $O$, respectively. Define the shorthand 
    \begin{equation} \label{eq:alpha'}
        \alpha' \coloneqq \frac{2\beta}{(\lambda_{\max}(O) -\lambda_{\min}(O))},
    \end{equation}
   and note that $\alpha' \in (0,1/2)$ with the condition assumed on $\beta$. With that, let us consider the following two states:
   \begin{align}
       \rho_0 & \coloneqq \left( \frac{1}{2} + \alpha '\right) |\psi_{\max} \rangle \! \langle \psi_{\max}| +   \left( \frac{1}{2} - \alpha '\right) |\psi_{\min} \rangle \! \langle \psi_{\min}|  \label{eq:rho_0}\\
       \rho_1 & \coloneqq \left( \frac{1}{2} -\alpha '\right) |\psi_{\max} \rangle \! \langle \psi_{\max}| +   \left( \frac{1}{2} + \alpha '\right) |\psi_{\min} \rangle \! \langle \psi_{\min}| \label{eq:rho_1}. 
   \end{align}
   We also have that 
   $T(\rho_0, \rho_1) =2 \alpha'$ since 
   \begin{align}
       \| \rho_0- \rho_1\|_1 &=2 \alpha' \| |\psi_{\max} \rangle \! \langle \psi_{\max}|-|\psi_{\min} \rangle \! \langle \psi_{\min}| \|_1 =4\alpha',
   \end{align}
   where the last equality holds since these (normalized) eigenvectors are orthogonal by the assumption $\lambda_{\max}(O) \neq \lambda_{\min}(O)$.
   Furthermore, we have that
   \begin{align}
       \Tr[O \rho_0] & = \left( \frac{1}{2} + \alpha '\right) \lambda_{\max}(O) + \left( \frac{1}{2} - \alpha '\right) \lambda_{\min}(O) \\ 
      \Tr[O \rho_1] & = \left( \frac{1}{2} - \alpha '\right) \lambda_{\max}(O) + \left( \frac{1}{2} + \alpha '\right) \lambda_{\min}(O) ,  
   \end{align}
   leading to 
\begin{align}
    \Tr[O\rho_0] -\Tr[O \rho_1] &= 2 \alpha' (\lambda_{\max}(O)- \lambda_{\min}(O)) 
    =4 \beta,
\end{align}
after substituting for $\alpha'$.
Also, 
\begin{align}
     \Tr[O\rho_0] +\Tr[O \rho_1] &=\lambda_{\max}(O)+ \lambda_{\min}(O), \\
     \lambda_{\max}(O)+ \lambda_{\min}(O) &= 2  \Tr[O\rho_0] -4 \beta =2  \Tr[O\rho_1] +4 \beta. \label{eq:rel_trO_rho_lambda}
\end{align}

Let us consider the following hypothesis testing setting: Two states $\rho_0$ and $\rho_1$ with equal prior probability of occurrence and we define 
\begin{equation}
    H_0 : \rho = \rho_0  \quad H_1: \rho= \rho_1.
\end{equation}
Assume that we have a protocol that, having access to $n$ copies of $\cA(\rho)$, can estimate $\Tr[O\rho]$ up to an accuracy $\beta$ with probability at least $1-\eta$.
Given $n$ copies of $\cA(\rho)$  where $\rho \in \{\rho_0,\rho_1\}$, run the estimator and threshold it as follows:
Let $\hat{E}_\rho$ be the estimate obtained and define the following testing criterion:
\begin{equation}
T(\hat{E}_\rho ) = \begin{cases}
    0   & \hat{E}_\rho  \geq \frac{1}{2} \left(\lambda_{\max}(O)+ \lambda_{\min}(O)\right)\\
    1  &  \textnormal{otherwise.}
\end{cases}
\end{equation}
Now, we study the error probability of this protocol. Under $H_0$, we have $\rho=\rho_0$. With that, an error happens iff we declare the state given to us is $\rho_1$ in the testing phase, i.e.; we get the test outcome to be one. This leads to the following: 
\begin{align}
    \operatorname{Pr}[\operatorname{error} | H_0] &=  \operatorname{Pr}[ T(\hat{E}_{\rho_0}  )=1] \quad 
    \\
    &=  \operatorname{Pr}\!\left[ \frac{1}{2} \left(\lambda_{\max}(O)+ \lambda_{\min}(O)\right) > \hat{E}_{\rho_0}  \right] \\
    &= \operatorname{Pr}\!\left[ 
    \Tr[O\rho_0]- \hat{E}_{\rho_0} > 2 \beta \right] \\ 
     &  \leq \operatorname{Pr}\!\left[ 
    \Tr[O\rho_0]- \hat{E}_{\rho_0}   >  \beta \right] \\
   & \leq \eta,
\end{align}
where the third equality follows from~\eqref{eq:rel_trO_rho_lambda}; first inequality since $\Tr[O\rho_0]- \hat{E}_{\rho_0} > 2 \beta  \implies     \Tr[O\rho_0]- \hat{E}_{\rho_0}  >  \beta  $, and the last inequality since the estimation protocol has the following accuracy criterion: $\operatorname{Pr}\!\left[ 
    \left | \Tr[O\rho_0]- \hat{E}_{\rho_0}  \right | \leq  \beta \right] \geq 1-\eta $ leading to $\operatorname{Pr}\!\left[ 
      \Tr[O\rho_0]- \hat{E}_{\rho_0}    \leq  \beta \right] \geq 1-\eta $.

Similarly, under hypothesis $H_1$, we have 
\begin{align}
    \operatorname{Pr}[\operatorname{error} | H_1] &=  \operatorname{Pr}[ T(\hat{E}_{\rho_1}  )=0] \\
    &=  \operatorname{Pr}\!\left[ \frac{1}{2} \left(\lambda_{\max}(O)+ \lambda_{\min}(O)\right) \leq  \hat{E}_{\rho_1} \right] \\
    &= \operatorname{Pr}\!\left[ 
  \hat{E}_{\rho_1}   -\Tr[O\rho_1] > 2 \beta \right] \\ 
     &  \leq \operatorname{Pr}\!\left[ 
     \hat{E}_{\rho_1}   -\Tr[O\rho_1] >  \beta \right] \\
   & \leq \eta,
\end{align}
where the last inequality follows since $\operatorname{Pr}\!\left[ 
    \left | \Tr[O\rho_1]- \hat{E}_{\rho_1} \right | \leq  \beta \right] \geq 1-\eta $ leading to \\$\operatorname{Pr}\!\left[ 
  \hat{E}_{\rho_1}    -  \Tr[O\rho_1]  \leq  \beta \right] \geq 1-\eta $.
With these, we have that the average error of the hypothesis testing protocol is at most $\eta$ since 
\begin{equation}
     \operatorname{Pr}[\operatorname{error}] =\frac{1}{2} \operatorname{Pr}[\operatorname{error} | H_0] + \frac{1}{2} \operatorname{Pr}[\operatorname{error} | H_1] \leq \eta.
\end{equation}
However, to achieve this error tolerance, by the lower bounds for private hypothesis testing we need $n$ to satisfy the following lower bound due to~\cref{thm:bounds_sample_C_private}: for $\eta \in (0,1/4)$ since $p=q=1/2$ 
\begin{equation}
    n \geq \frac{\ln\! \left( \frac{1}{4 \eta (1-\eta)}\right) e^\varepsilon }{ 2(e^\varepsilon -1)^2  \left[T(\rho_0,\rho_1)\right]^2}
\end{equation}

Since $T(\rho_0,\rho_1) =2\alpha'$, we arrive at the following:
\begin{equation}
    n \geq \frac{\ln\! \left( \frac{1}{4 \eta (1-\eta)}\right) e^\varepsilon  (\lambda_{\max}(O) -\lambda_{\min}(O))^2}{ 32(e^\varepsilon -1)^2  \beta^2 }.
\end{equation}
Since the above holds for all private mechanisms $\cA \in \cB^\varepsilon$ and protocol $\cP_O$, we arrive at the desired lower bound in the first inequality of the proposition. 

Now for the regime $\varepsilon \in (0,1)$, we have that \footnote{To see why the~\eqref{eq:eps_le_1} holds, consider $g(\varepsilon) \coloneqq 2 \varepsilon^2 e^\varepsilon -(e^\varepsilon -1)^2$ and $g(\varepsilon) >0$ for $\varepsilon \in (0,1)$. This is because the first derivative $g'(\varepsilon) =2e^\varepsilon ( (\varepsilon +1)^2 -e^\varepsilon) \geq 0$ leading to $g(\varepsilon)$ being a non-decreasing function of $\varepsilon$ with $e^\varepsilon \leq 1+ 2 \varepsilon + \varepsilon^2$, so that $g(\varepsilon) > g(0)=0$ for $\varepsilon \in (0,1)$}
\begin{equation} \label{eq:eps_le_1}
    \frac{e^\varepsilon }{(e^\varepsilon -1)^2} > \frac{1}{2 \varepsilon^2}, 
\end{equation}
and together with $\eta \in (0,1/4)$, we conclude the second desired statement by substituting into the first inequality of the Proposition.

\subsection{Proof of \Cref{lem:A_O^q_private}}
\label{app:lemma_A_O^q_private_proof}
    Let $\rho, \sigma \in \cD(\cH)$, and $M_{AB}$ such that $0 \leq M_{AB} \leq I_{AB}$, with $A$ and $B$ denoting the output space of $ \cA_q \!\left( \cM_P (\cdot)\right) $ and the classical register having $P$, respectively. 
    Then, consider, 
    \begin{align}
        \Tr\!\left[ M_{AB} \  \cA_q^O(\rho) \right] &=  \Tr\!\left[ M_{AB}   \sum_{P} p(P) \ \cA_q \!\left( \cM_P (\rho)\right) \otimes |P \rangle\!\langle P| \right] \\
        &= \sum_{P} p(P) \Tr\!\left[ M_{AB} \cA_q \!\left( \cM_P (\rho)\right) \otimes |P \rangle\!\langle P| \right] \\
        &= \sum_{P} p(P) \Tr\!\left[ \langle P| M_{AB} | P\rangle \  \cA_q \!\left( \cM_P (\rho)\right)  \right] \\ 
        &=\sum_{P} p(P) \Tr\!\left[  M'_P  \cA_q \!\left( \cM_P (\rho)\right)  \right] \\
        & \leq \sum_{P} p(P) \left( e^{\varepsilon} \Tr\!\left[  M_P'  \cA_q \!\left( \cM_P (\sigma)\right)\right] +\delta  \right)  \\
        &= e^\varepsilon \Tr\!\left[ M_{AB} \  \cA_q^O(\sigma) \right] + \delta,
    \end{align}
where $M_P' \coloneqq \langle P| M_{AB} | P\rangle $ in the fourth equality; first inequality follows since $0 \leq M'_P \leq I$ and $\cA_q\circ \cM_P$ satisfies ($\varepsilon,\delta$)-QLDP with $q= 2(1-\delta)/ (e^\varepsilon +1)$ by substituting $d=2$ in~\cref{prop:optimal_P_QLDP_Dep} (see also~\cite[Proposition~1]{nuradha2024contraction} for $\delta=0$), and the last equality by rearranging the terms and redoing the equalities in the reverse order to get the desired expression. 
Since the above inequality holds for all pairs of states $\rho, \sigma$ and measurement operators $M_{AB}$, we conclude that $\cA_q^O \in \cB^{\varepsilon,\delta}$.

\subsection{Proof of \Cref{prop:upper_1}}
\label{app:prop_upper_bound_estimate_1}

Assume that we are given $n$ i.i.d.\ classical outputs of $\mathcal A_q^O(\rho)$.
For each sample $i\in\{1,\dots,n\}$, let the observed pair be $(Y_i, P_i)$, where $Y_i\in\{0,1\}$ is the computational-basis bit obtained after applying the channels $\cA_q \circ \cM_{P_i}$ and $P_i\sim p(\cdot)$ is the Pauli label. Using the pair $(Y_i, P_i)$, we define the random variables $X_i \coloneqq  (-1)^{Y_i}\in\{+1,-1\} $ and
\begin{equation} \label{eq:def_Z_i}
    Z_i \coloneqq \frac{S}{1-q}\,\operatorname{sgn}(\alpha_{P_i}) X_i.
\end{equation}

Let the estimator be $ \widehat{E} \coloneqq \frac{1}{n}\sum_{i=1}^n Z_i$
and note that $\widehat{E}$ is a random variable with $\mathbb{E} [\widehat{E}]= \mathbb{E}[Z_i]$.
Furthermore, $ \mathbb{E}[Z_i |P_i=P]= \frac{S}{1-q}\,\operatorname{sgn}(\alpha_{P}) \mathbb{E}[X_i |P_i=P].$
We also have that $\mathbb{E}[X_i |P_i=P]  
    = \operatorname{Pr}(Y_i=0|P_i=P) - \operatorname{Pr}(Y_i=1|P_i=P).$
By considering the action of the channel $\cA_q \circ \cM_P$ and recalling that for a qubit state $\omega$ $\cA_q(\omega)= (1-q) \omega  + q I/2$, we have that 
\begin{align}
   \operatorname{Pr}(Y_i=0|P_i=P) & = \Tr\!\left[ \left(\frac{I+P}{2}\right) \rho \right] \left(1-q + \frac{q}{2} \right) + \frac{q}{2} \Tr\!\left[ \left(\frac{I-P}{2}\right) \rho \right] \\
   &= \frac{1}{2} +\frac{(1-q)}{2} \Tr[P \rho] 
\end{align} and $ 
   \operatorname{Pr}(Y_i=1|P_i=P)  = \frac{1}{2} -\frac{(1-q)}{2} \Tr[P \rho]$,
leading to $ \mathbb{E}[X_i |P_i=P]=  (1-q)\Tr[P \rho]. $
With those ingredients, consider 
\begin{align}
    \mathbb{E} [\widehat{E}] &= \mathbb{E}[Z_i] \\
    &= \sum_P p(P) \  \mathbb{E}[Z_i |P_i=P] \\
    &= \sum_P p(P) \frac{S}{1-q}\,\operatorname{sgn}(\alpha_{P}) \  \mathbb{E}[X_i |P_i=P] \\ 
    &= \sum_P p(P) \frac{S}{1-q}\,\operatorname{sgn}(\alpha_{P}) \  (1-q)\Tr[P \rho] \\
    &= \sum_P \frac{|\alpha_P|}{S}  \times  S \operatorname{sgn}(\alpha_{P}) \ \Tr[P \rho] \\
    &= \sum_P \alpha_P \Tr[P \rho] \\
    &= \Tr[ O \rho], 
\end{align}
by recalling that $O= \sum_P \alpha_P P$.

By~\eqref{eq:def_Z_i}, we have that $Z_i \in [a, b]$
with $a = -S /(1-q)$ and $b=S /(1-q)$. Then by using $\mathbb E[\widehat{E}]=\Tr[O\rho]$ and Hoeffding's inequality yields
\begin{equation} \label{eq:Accuracy_gurantee_Hoeff}
    \Pr\!\left[\left|\widehat{E}-\Tr[O\rho] \right|\geq \beta \right]
\leq
2\exp\!\left(-\frac{n\beta^2(1-q)^2}{2S^2}\right),
\end{equation}
and it suffices to choose $n$ such that the right hand side of the above inequality is at most $\eta$, i.e.
\begin{equation}
n \geq \frac{2 S^2}{\beta^2(1-q)^2} \ln\!\left(\frac{2}{\eta}\right).
\end{equation}
Since $q=2(1-\delta)/(e^\varepsilon +1)$, we arrive at the desired result together with~\eqref{eq:Accuracy_gurantee_Hoeff}.

 For $\varepsilon \in (0,1)$ and $\delta=0$, 
   \begin{equation} \label{eq:epa_one_less_upper}
     \frac{1}{1-q}   = \frac{e^\varepsilon +1}{e^\varepsilon -1} 
         \leq \frac{4}{\varepsilon},
    \end{equation}
    due to $(e^\varepsilon -1)/(e^\varepsilon +1)= \tanh(x/2)$  and $\tanh(x/2) \geq x/4$ for $x\in[0,1]$.

With that for $\varepsilon \in (0,1)$, by choosing 
\begin{equation}
n \geq \frac{32 \ S^2}{\beta^2 \varepsilon^2} \ln\!\left(\frac{2}{\eta}\right),
\end{equation}
the described protocol achieves the demanded accuracy criterion by~\eqref{eq:Accuracy_gurantee_Hoeff}, so that providing an upper bound on optimal $n^*$, concluding the proof.

\subsection{Proof of \Cref{prop:upper_2}} \label{app:classical_shadow_SC_improved}


Recall that $\cS_{\hat{p},d}$ is equivalent to a depolarizing channel (see \cite[Claim~4.10]{Koh2022classicalshadows}) with flipping parameter $q=1-(1-p)/(d+1)$.
 By \Cref{prop:optimal_P_QLDP_Dep}, if $q \geq d(1-\delta)/(e^\varepsilon +d -1)$, $\cS_{\hat{p},d}$ satisfies $(\varepsilon, \delta)$-QLDP. Let us choose 
\begin{equation} \label{eq:hat_p_choice}
     \hat{p} = 1-\min \left\{1, \frac{(e^\varepsilon -1 +d \delta) (d+1)}{e^\varepsilon +d-1}\right\}.
\end{equation}
Then we have that $\hat{p} \in [0,1]$. We also see that $\hat{p}=0$ when $e^\varepsilon +\delta (d+1) \geq 2$. To this end, $q=1-1/(d+1)$ and we have that $q \geq d(1-\delta)/(e^\varepsilon +d -1)$, ensuring that $\cS_{\hat{p}=0,d}$ is $(\varepsilon,\delta)$ satisfies $(\varepsilon,\delta)$-QLDP when $e^\varepsilon +\delta (d+1) \geq 2$.

Similarly, when $e^\varepsilon +\delta (d+1) \leq 2$, we have $\hat{p}= p^* \coloneqq 1- \frac{(e^\varepsilon -1 +d \delta) (d+1)}{e^\varepsilon +d-1}$. With that choice we get $q=d(1-\delta)/(e^\varepsilon+d-1)$, showcasing that $\cS_{\hat{p}=p^*,d}$ satisfies $(\varepsilon,\delta)$-QLDP in the other regime as well. Thus, we conclude that $\cS_{\hat{p},d}$ with the choice $\hat{p}$ in~\eqref{eq:hat_p_choice}, satisfies $(\varepsilon,\delta)$-QLDP. 



\subsection{Pauli and Clifford Group}
For a qubit system ($d=2$), Pauli operators are the collection of $\{X, Z, Y, I\}$, where 
\begin{align}
    X \coloneqq \begin{bmatrix}
0 & 1 \\
1 & 0
\end{bmatrix}, \quad Z =
\begin{bmatrix}
1 & 0 \\
0 & -1
\end{bmatrix}, \quad Y =
\begin{bmatrix}
0 & -i \\
i & 0
\end{bmatrix}, \quad I =
\begin{bmatrix}
1 & 0 \\
0 & 1
\end{bmatrix}.
\end{align}
Also we write $\mathcal{P}_1 = \{ I, X, Y, Z \}.$

The Pauli group on $m$-qubits ($d=2^m$) is generated by tensor products of single-qubit Pauli matrices.
An $m$-qubit Pauli operator $P \in \cP_m$ has the form
\begin{equation}
    P = \sigma_1 \otimes \sigma_2 \otimes \cdots \otimes \sigma_m,
\end{equation}
where  each $\sigma_j \in \{ I, X, Y, Z \}$.  In fact, we have 
\begin{equation} \label{eq:pauli_operators}
    \cP_m \coloneqq \{ I, X, Y, Z \}^{\otimes m}. 
\end{equation}

The Clifford group $\cC_m$ is the set of unitary operations that map Pauli
operators to Pauli operators under conjugation:
\begin{equation} \label{eq:clifford_group}
    \mathcal{C}_m
= \{ U \in \mathbb{U}(2^m) :
\forall P \in \mathcal{P}'_m \ \exists  \widetilde{P}  \in \mathcal{P}'_m \  \textnormal{ s.t. } U P U^\dagger = \widetilde{P} \},
\end{equation}
where $\mathbb{U}(2^m)$ denotes the group of unitary matrices of size $2^m$, and \\ 
$\mathcal{P}'_m \coloneqq \left\{ i^k \ \sigma_1 \otimes \cdots \otimes \sigma_m
:
k \in \{0,1,2,3\},\ \sigma_j \in \{I,X,Y,Z\}
\right\}$.
\end{document}